\documentclass[11pt,a4paper]{article}
\pdfoutput=1

\usepackage[margin=2.7cm]{geometry}

\usepackage{amsmath,amssymb,amsthm}
\usepackage{booktabs,doi,graphicx,latexsym,url,color,xcolor,xspace,microtype,hyperref,euscript}
\usepackage[numbers,sort&compress]{natbib}
\usepackage{multirow}

\usepackage[inline]{enumitem}
\setlist[itemize]{label=--}
\setlist[enumerate]{label=(\arabic*),labelindent=\parindent,leftmargin=*}

\newtheorem{theorem}{Theorem}
\newtheorem{lemma}[theorem]{Lemma}
\newtheorem{corollary}[theorem]{Corollary}
\newtheorem{definition}[theorem]{Definition}
\theoremstyle{remark}

\theoremstyle{definition}

\newcommand{\congest}{\ensuremath{\mathsf{CONGEST}}\xspace}
\newcommand{\cc}{\ensuremath{\mathsf{CC}}\xspace}
\newcommand{\rcc}{\ensuremath{\mathsf{RCC}}\xspace}
\newcommand{\disj}{\ensuremath{\mathsf{DISJ}}\xspace}

\newcommand{\A}{\EuScript{A}}
\newcommand{\C}{\EuScript{C}}
\newcommand{\D}{\EuScript{D}}
\newcommand{\F}{\EuScript{F}}
\newcommand{\G}{\EuScript{G}}

\newcommand{\degen}{d}
\newcommand{\vc}{\tau}
\DeclareMathOperator{\tw}{tw}

\newcommand{\mcistoC}[1]{C_{#1,\to C}}
\newcommand{\mcistoI}[1]{C_{#1,\to I}}

\definecolor{citecolor}{HTML}{0000C0}
\definecolor{urlcolor}{HTML}{000080}
\hypersetup{
    colorlinks=true,
    linkcolor=black,
    citecolor=citecolor,
    filecolor=black,
    urlcolor=urlcolor,
}

\usepackage{todonotes}

\newenvironment{myabstract}
{\list{}{\listparindent 1.5em%
		\itemindent    \listparindent
		\leftmargin    1cm
		\rightmargin   1cm
		\parsep        0pt}%
	\item\relax}
{\endlist}

\newenvironment{mycover}
{\list{}{\listparindent 0pt
		\itemindent    \listparindent
		\leftmargin    1cm
		\rightmargin   1cm
		\parsep        0pt}%
	\raggedright
	\item\relax}
{\endlist}

\newcommand{\myemail}[1]{\,$\cdot$\, {\small #1}}
\newcommand{\myaff}[1]{\,$\cdot$\, {\small #1}\par\medskip}

\begin{document}

\begin{mycover}
	{\LARGE\bfseries\boldmath Beyond Distributed Subgraph Detection: Induced Subgraphs, Multicolored Problems and Graph Parameters \par}
    
	\bigskip
	\bigskip    
    \textbf{Janne H.\ Korhonen}
    \myemail{janne.korhonen@ist.ac.at}
    \myaff{IST Austria}
    
    \textbf{Amir Nikabadi}
    \myemail{amir.nikabadi@ens-lyon.fr}
    \myaff{ENS de Lyon}
\end{mycover}

\medskip

\begin{myabstract}
\noindent\textbf{Abstract.}
Subgraph detection has recently been one of the most studied problems in the \congest model of distributed computing. In this work, we study the distributed complexity of problems closely related to subgraph detection, mainly focusing on \emph{induced subgraph detection}. The main line of this work presents lower bounds and parameterized algorithms w.r.t \emph{structural parameters} of the input graph:
\begin{itemize}
    \item On general graphs, we give unconditional lower bounds for induced detection of cycles and patterns of treewidth $2$ in \congest. Moreover, by adapting reductions from centralized parameterized complexity, we prove lower bounds in \congest for detecting patterns with a $4$-clique, and for induced path detection conditional on the hardness of triangle detection in the congested clique.
    \item On graphs of bounded degeneracy, we show that induced paths can be detected fast in \congest using techniques from parameterized algorithms, while detecting cycles and patterns of treewidth $2$ is hard.
	\item On graphs of bounded vertex cover number, we show that induced subgraph detection is easy in \congest for any pattern graph. More specifically, we adapt a centralized parameterized algorithm for a more general \emph{maximum common induced subgraph detection} problem to the distributed setting.
    \end{itemize}
In addition to these induced subgraph detection results, we study various related problems in the \congest and congested clique models, including for \emph{multicolored} versions of subgraph-detection-like problems.
\end{myabstract}
\clearpage


\section{Introduction}

\emph{Subgraph detection} is one of the most studied problems in the \congest and congested clique models of distributed computing~\cite{tritri, censor2015algebraic, korhonen2017deterministic,drucker13,disc2017property,fraigniaud2017spaa,fischer2018possibilities,partition-trees,censor2021tight}. The complexity of distributed subgraph detection is understood for many pattern graphs -- for example, in the \congest model, tight bounds are known for \emph{path}~\cite{korhonen2017deterministic,disc2017property} and \emph{odd cycle} detection~\cite{korhonen2017deterministic,drucker13}, and it is known that pattern graphs requiring almost quadratic time exist~\cite{fischer2018possibilities}. However, unresolved questions remain about the exact complexity of, e.g., \emph{triangle} detection in either \congest or congested clique, and \emph{even cycle} detection in \congest.

In this work, we look at the closely related \emph{induced subgraph detection} problem, which has so far not received any attention in the distributed setting.
In particular, we aim to understand the complexity of induced subgraph detection for common pattern graphs, such as paths and cycles, as well as how the situation contrasts with the non-induced case. It is well known that in the centralized setting, induced subgraph detection is generally more difficult than non-induced subgraph detection, so one would expect that situation is the same also in the distributed setting. 

\subsection{Background and setting}

Before presenting our results, we start by discussing the wider context of distributed subgraph detection problems. As mentioned above, we work in the \congest and congested clique models of distributed computing and use $G$ and $n$ to denote the input graph and the number number of nodes in the input graph, respectively.

In the paper, we mostly consider subgraph detection and induced subgraph detection problems; we are given a pattern graph $H$ with $k$ nodes, known to all nodes in $G$, and the task is to decide if the input graph $G$ contains $H$ as a subgraph or an induced subgraph; more precisely, any node $v$ that is part of an admissible copy of $H$ should report that the input is a yes-instance.

\paragraph{Fixed-parameter tractability.} Subgraph and induced subgraph detection problems can be viewed as \emph{parameterized problems}; such problems are studied in centralized setting under the field of \emph{parameterized complexity}~\cite{cygan2015parameterized}. A parameterized problem is defined by the input and a problem parameter $k$ -- formally, a \emph{(complexity) parameter $k$} is a mapping from the input instance to natural numbers. The basic question of centralized parameterized complexity is to understand which problems are \emph{fixed-parameter tractable}, i.e. have algorithms with running time $f(k) |x|^{O(1)}$, where $f$ is an arbitrary function and $x$ is the binary encoding of the input instance. For example, $k$-cycle detection can be viewed as a parameterized problem.

Similarly, one can consider fixed-parameter tractability in the distributed setting. The strictest definition is to ask which problems have distributed algorithms where the running time depends only on the parameter $k$~\cite{siebertz2019parameterized,ben2018parameterized}. However, this arguably does not capture all fixed-parameter tractability phenomena in distributed models -- e.g. $k$-cycle detection cannot be solved in $f(k)$ rounds for any function $f$ in the \congest model.

A more general perspective is to ask what is the smallest function $T$ such that a parameterized problem can be solved in $f(k) \cdot T(n)$ rounds,
for some function $f \colon \mathbb{N} \to \mathbb{N}$. Several results of this type are known for subgraph detection problems; for example, $k$-cycle detection can be solved in $O(k2^k n)$ rounds in the \congest model~\cite{korhonen2017deterministic,disc2017property}, and in $2^{O(k)}n^{0.158}$ rounds in the congested clique model~\cite{censor2015algebraic}, though these bounds are not tight for \emph{even-length} cycles~\cite{partition-trees,fischer2018possibilities}. 

\paragraph{Parameters and graph structure.} For subgraph and induced subgraph detection problems, the natural complexity parameter is the number of nodes $k$ in the pattern graph. However, parameterized complexity frequently studies other complexity parameters -- for our purposes, the most relevant are structural graph parameters, in particular degeneracy $\degen(G)$,
treewidth $\tw(G)$, and vertex cover number $\vc(G)$ (see Section~\ref{sec:preliminaries} for the precise definitions). While bounded degeneracy (equivalently, bounded arboricity)~\cite{barenboim2010sublogarithmic,korhonen2017deterministic,coloringbook}
has been studied in the distributed setting, bounded treewidth and bounded vertex cover number less so.

Given a structural parameter $p$, we can consider the complexity of subgraph or induced subgraph detection parameterized either by the structural parameter $p(G)$ of the input graph, or by the structural parameter $p(H)$ of the pattern graph. Note that we have
\[  \degen(G) \le \tw(G) \le \vc(G)\,.\]
For parameters $p_1$ and $p_2$ with $p_1(G) \le p_2(G)$, upper bounds w.r.t.\ parameter $p_2$ imply upper bounds w.r.t.\ parameter $p_1$, and lower bounds w.r.t.\ parameter $p_1$ imply lower bounds w.r.t. parameter $p_2$. 

\paragraph{Lower bounds and reductions.} The standard technique for proving unconditional \congest lower bounds is by reduction from communication complexity problems, most often using \emph{families of lower bound graphs}~\cite{dassarma12,drucker13,CHKP17,AbboudCK16,frischknecht2012,czumaj2019detecting} (see Section~\ref{sec:preliminaries}). By contrast, reductions between problems are less useful in the \congest model, as the model can implement only very limited reductions efficiently.

However, there are still uses for reductions in distributed complexity theory, which we will apply in this work. First, in the congested clique, sub-polynomial round reductions can be used to establish relative complexities of problems~\cite{korhonen2018towards}. Second, as noted by Bacrach et al.~\cite{bacrach2019hardness}, centralized reductions can be used to transform families of lower bound graphs for one problem into families of lower bound graphs for a second problem.

\begin{table}[!t]
\newcommand{\myitem}{\ \ $\cdot$ }
\centering
\caption{Lower bounds on general graphs. Improved lower bounds of Le Gall and Miyamoto~\cite{legall2021induced} are independent and concurrent work (see main text.)}\label{tab:results-lb}
\vspace{0.2em}
\begin{tabular*}{\linewidth}{@{}l@{\extracolsep{\fill}}l@{}l@{}r@{}}
\toprule
Problem & Bound &  & \\
\midrule
Induced $2k$-cycle ($k \ge 3$) & $\Omega(n/\log n)$ & & Section~\ref{sec:induced-even-cycle-lb} \\
Induced $H$-detection &  &  & \\
\myitem any $H$ with $4$-clique    & $\Omega(n^{1/2}/\log n)$ & & Section~\ref{sec:lb-pattern-with-clique} \\
\myitem some $H$ with $\tw(H) = 2$${}^\dagger$ & $\Omega(n^{2-\varepsilon})$ & & Section~\ref{sec:lb-low-treewidth} \\
\midrule
Multicolored $k$-cycle ($k \ge 4$)        &  $\Omega(n/\log n)$ & & Section~\ref{sec:multicolor-lb} \\
Multicolored induced path of length $k$ ($k \ge 6$) &  $\Omega(n/\log n)$ & & Section~\ref{sec:multicolor-lb} \\
\midrule
Induced $k$-cycle ($k \ge 4$) & $\tilde{\Omega}(n)$ & & \cite{legall2021induced} \\
Induced $k$-cycle ($k \ge 8$) & $\tilde{\Omega}(n^{2-\Theta(1/k)})$ & & \cite{legall2021induced} \\

\bottomrule
\multicolumn{4}{l}{\ \footnotesize{${}^\dagger$holds for any $\varepsilon > 0$, for some $H$ that is chosen depending $\varepsilon$}}
\end{tabular*}
\vspace{0.5em}
\caption{Bounds w.r.t.~structural graph parameters. Results attributed to \cite{korhonen2017deterministic} follow directly from the proofs in that work, but are not stated in that work for induced subgraphs.}\label{tab:results-param-ub}
\vspace{0.2em}
\begin{tabular*}{\linewidth}{@{}l@{\extracolsep{\fill}}l@{}l@{}r@{}}
\toprule
Problem & Bound & & \\
\midrule
Induced $k$-tree${}^\dagger$  & $2^{O(k\degen(G))} k^k + O(\log n)$ & & Section \ref{sec:induced-trees-degen} \\
(Induced) $H$-detection, & \multirow{2}{*}{$\Omega(n^{1-\varepsilon})$} &\multirow{2}{*}{holds for $\degen(G) = 2$}& \multirow{2}{*}{Section \ref{sec:lb-low-treewidth-degen}}\\
\ \ some $H$ with $\tw(H) = 2$${}^\ddagger$  &&& \\
\midrule
(Induced) $k$-cycle ($k \ge 6$)  & $\Omega(n^{1/2}/\log n)$ & holds for $\degen(G) = 2$& \cite{korhonen2017deterministic} \\
Induced $4$-cycle & $O(\degen(G) + \log n)$ && \cite{korhonen2017deterministic} \\
Induced $5$-cycle & $O(\degen(G)^2 + \log n)$ && \cite{korhonen2017deterministic} \\
\midrule
MCIS    & $2^{O(\vc^2)}$ & $\vc = \vc(G) + \vc(H)$ & Section \ref{sec:mcis} \\
Induced subgraph & $2^{O((\vc(G) + k)^2)}$ & & Section \ref{sec:mcis} \\
\bottomrule
\multicolumn{4}{l}{\ \footnotesize{${}^\dagger$randomized algorithm, can be derandomized with extra assumptions and worse running time}}\\
\multicolumn{4}{l}{\ \footnotesize{${}^\ddagger$holds for any $\varepsilon > 0$, for some $H$ that is chosen depending $\varepsilon$}}
\end{tabular*}
\end{table}

\subsection{Results: induced subgraph detection on general graphs}

First, we consider the hardness of induced subgraph detection on general graphs. We show that for common pattern graphs, the induced version of the problem is at least as hard as the non-induced version, and in many cases harder.

\paragraph{Unconditional lower bounds.} We start with unconditional lower bounds for induced subgraph detection in \congest; see Table~\ref{tab:results-lb} for a summary of these results.

For cycles of length at least $6$, we show that the induced cycle detection problem requires at least $O(n/\log n)$ rounds in the \congest model. The result follows from a combination of the existing lower bound construction for odd-length cycles, and new construction for induced even cycles. By comparison, the existing lower bounds for non-induced subgraph detection in \congest are $\Omega(n^{1/2}/\log n)$ for even cycle detection~\cite{korhonen2017deterministic}, and $\Omega(n/\log n)$ for odd cycle detection excluding triangles~\cite{drucker13}; it is also known that even cycles can be detected in $O(n^\delta)$ time, for $\delta < 1$ that depends on the length of the cycle~\cite{fischer2018possibilities}.

We also prove that there are pattern graphs for which induced subgraph detection (and also non-induced detection) requires near-quadratic time in \congest, in a similar spirit as the hard pattern graphs for non-induced subgraph detection presented by Fischer et al.~\cite{fischer2018possibilities}. Moreover, we show that these pattern graphs can be constructed to have treewidth $2$; contrast this with the centralized setting, where low-treewidth patterns are easy to detect~\cite{alon1995color}.

\paragraph{Unconditional lower bounds: recent independent work.} After submitting the first version of this paper, we learned about the independent and concurrent work of \citet{legall2021induced}, which gives lower bounds for induced cycle detection and diamond listing. In particular, they show that detecting induced $k$-cycles requires $\tilde{\Omega}(n)$ rounds for any $k \ge 4$, and $\tilde{\Omega}(n^{2-\Theta(1/k)})$ rounds for any $k \ge 8$. These results subsume our lower bounds for induced cycle and treewidth-2 subgraph detection. 

\paragraph{Reductions.} Next, we turn our attention to conditional lower bounds for problems where standard \congest lower bound techniques do not immediately yield unconditional lower bounds. See Figure~\ref{fig:problem-relations} for a summary of these results.

We adapt a centralized reduction of Dalirrooyfard et al.~\cite{10.1145/3313276.3316329} between clique and independent set detection and induced subgraph detection. Specifically, they show that detecting an induced subgraph $H$ that contains a $k$-clique ($k$-independent set) is as hard as detecting $k$-clique ($k$-independent set, resp.). We show that this reduction can also be implemented in the congested clique model.

It follows that detecting induced paths of length at least $5$ in either the \congest or congested clique model is at least as hard \emph{triangle detection} in the congested clique model, and more generally, detecting paths of length at least $2k-1$ in \congest or congested clique is as hard as detecting $k$-cliques in the congested clique.  By comparison, the best known upper bounds in the congested clique are $O(n^{0.158})$ for triangle detection~\cite{censor2015algebraic}, and $O(n^{1-1/k})$ for $k$-clique detection~\cite{tritri}; while no lower bounds for the congested clique model are known, improving over the $O(n^{0.158})$-round matrix multiplication based triangle detection would have major implications for distributed algorithms. However, it is worth noting that induced paths of length $2$ can be detected in $O(1)$ rounds in \congest, in contrast to triangles (see Appendix~\ref{sec:induced-short-paths}).

Moreover, the reduction allows us to lift the $\Omega(n^{1/2}/\log n)$ \congest lower bound of Czumaj and Konrad~\cite{czumaj2019detecting} for $4$-clique detection to induced and non-induced detection of any pattern graph $H$ that contains a $4$-clique.

\begin{figure}[!t]
\begin{center}
\includegraphics[width=\textwidth]{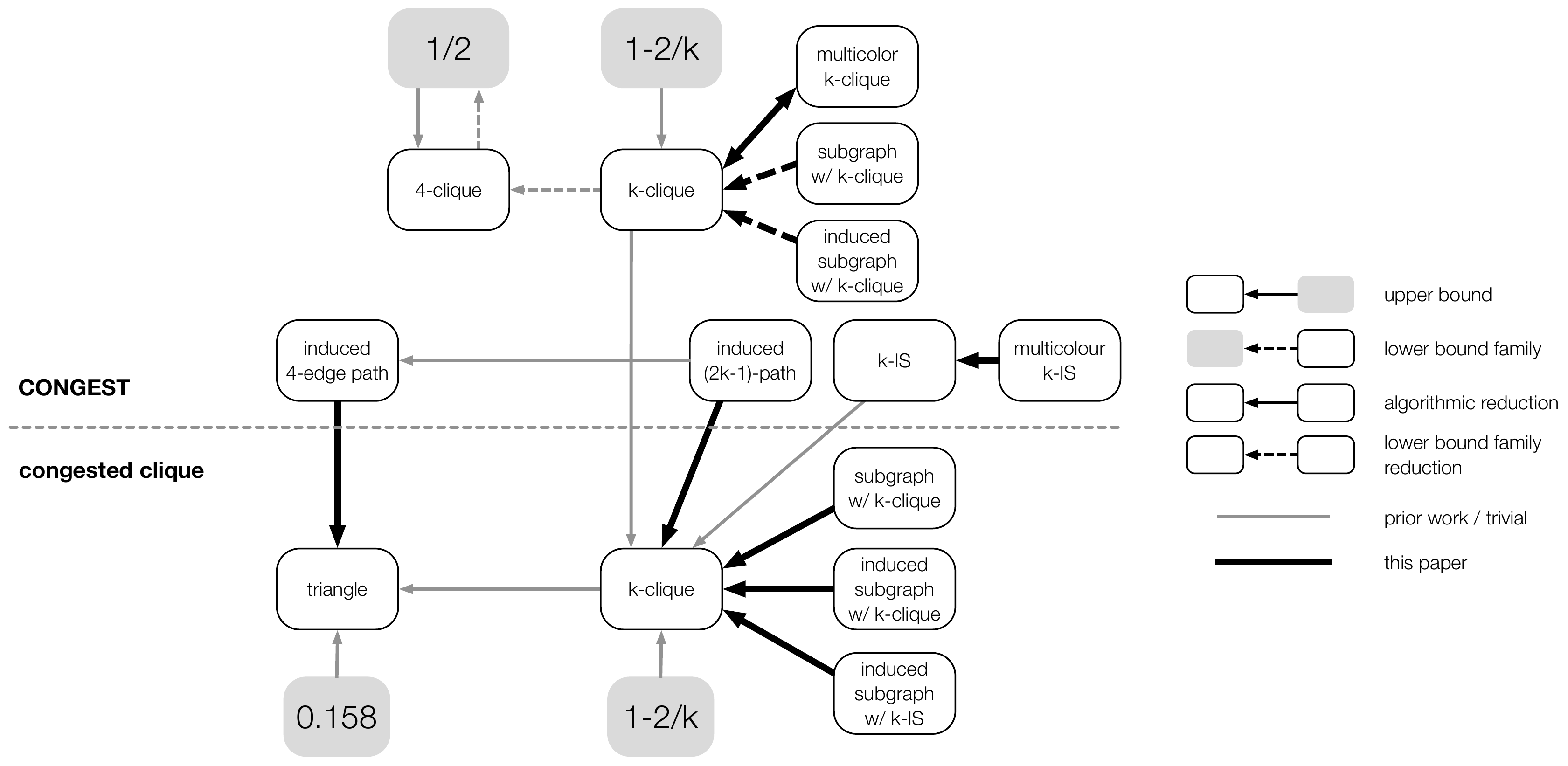}
\end{center}
\caption{Relationships between problems in \congest and congested clique. Results hold for any sufficiently large constant $k$. \emph{Upper bound} indicates an $\tilde{O}(n^\delta)$ round algorithm for the problem for specified $\delta$; \emph{lower bound family} indicates that there is a lower bound family giving $\tilde{\Omega}(n^\delta)$ lower bound for the problem for specified $\delta$; \emph{algorithmic reduction} from $P_1$ to $P_2$ indicates that an algorithm solving $P_2$ in $O(n^\delta)$ rounds implies the existence of an algorithm solving $P_1$ in $\tilde{O}(n^\delta)$ rounds, for any $\delta > 0$, and \emph{lower bound reduction} from $P_1$ to $P_2$ indicates that a lower bound family giving $\Omega(n^{\delta})$ lower bound for $P_1$ implies the existence of a lower bound family giving $\tilde{\Omega}(n^\delta)$ lower bound for $P_2$ for any $\delta > 0$. Notation $\tilde{O}$ and $\tilde{\Omega}$ hides polylogarithmic factors in $n$, as well as factors only depending on $k$, as we assume $k$ to be constant.}\label{fig:problem-relations}
\end{figure}

\paragraph{Multicolored problems.}

Finally, we consider \emph{multicolored} versions of subgraph detection tasks. In multicolored (induced) $H$-detection, we are given a labelling of the input graph $G$ with $k$ colors, and the task is to find a (induced) copy of $H$ that contains exactly one node of each color. Multicolored versions of problems have proven to be useful starting points for reductions in fixed-parameter complexity, and algorithms for a multicolored version of a problem can often be turned into an algorithm for the standard version via color-coding~\cite{alon1995color}.

We observe that multicolored versions of $k$-clique and $k$-independent set are closely related to their standard versions in the distributed setting, by adapting the simple centralized reductions to the distributed setting (see Figure~\ref{fig:problem-relations}). We then prove unconditional lower bounds of $\Omega(n / \log n)$ in \congest for multicolored versions of $k$-cycle detection, for $k \ge 4$, and for detection of induced paths of length $k$, for $k \ge 6$. These results imply that color-coding algorithms cannot be used directly to improve the state of the art for these problems --  for comparison, note that $k$-cycle detection can be solved in \congest in $o(n/\log n)$ rounds for even $k$, non-induced multicolored paths can be detected in $O(1)$ round in \congest, and we have no unconditional lower bounds for induced path detection. 

\subsection{Results: induced subgraph detection with structural parameters}

Next, we consider subgraph and induced subgraph detection tasks w.r.t.~structural graph parameters. We focus on the degeneracy $\degen(G)$ and the vertex cover number $\vc(G)$ of the input graph as the parameters in this section. See Table~\ref{tab:results-param-ub} for a summary of the results.

\paragraph{Bounded degeneracy.} We show that induced subgraph detection for any \emph{tree} on $k$ nodes can be solved in time $2^{O(k\degen(G))} k^k + O(\log n)$ rounds in \congest. As with the prior results on non-induced path, tree and cycle detection algorithms in \congest, this upper bound is based on centralized fixed-parameter algorithms, in this case using color-coding and random separation techniques~\cite{alon2007linear,cai2006random}.

On the lower bounds side, we show that there are treewidth $2$ pattern graphs that require near-linear time to detect as induced and non-induced subgraphs in \congest on input graphs of degeneracy $\degen(G) = 2$, via a slight modification of the proof for the general case discussed above. Note that any fixed pattern graph can be detected in $O(n)$ rounds when degeneracy is bounded, by having all nodes gather their distance-$k$ neighborhood.

For cycles, we note that results of Korhonen and Rybicki~\cite{korhonen2017deterministic} can be easily seen to imply that detecting induced $k$-cycles for $k \ge 6$ requires at least $\Omega(n^{1/2}/\log n)$ rounds to detect in \congest on graphs of degeneracy $\degen(G) = 2$, as well as that induced $4$-cycles can be detected in $O(\degen(G) + \log n)$ rounds, and induced $5$-cycles in $O(\degen(G)^2 + \log n)$ rounds.

\paragraph{Bounded vertex cover number.} For a more restrictive parameter than degeneracy, we consider induced subgraph detection parameterized by the vertex cover number $\vc(G)$ of the input graph. More precisely, we show a more general problem of \emph{maximum common induced subgraph} (\emph{MCIS}) can be solved fast; in this problem, we are given two graphs $G = (V_G, E_G)$ and $H = (V_H, E_H)$ as input, and the task is to find the maximum-size graph $G^*$ such that $G^*$ appears as an induced subgraph of both $G$ and $H$. In the distributed setting, we assume that $G$ is the input graph, and the second graph $H$ is known to every node.

In more detail, we show that a centralized branching algorithm from MCIS of Abu-Khzam et al.~\cite{ABUKHZAM201769} can be implemented in $2^{O((\vc(G) + \vc{H})^2)}$ rounds, i.e. without dependence on $n$, in the \congest model. This immediately implies that induced subgraph detection for any pattern graph $H$ on $k$ nodes can also be solved in $2^{O((\vc(G) + k)^2)}$ rounds.

\subsection{Additional related work}

\paragraph{Centralized subgraph and induced subgraph detection.} Subgraph detection has been widely studied in the centralized parameterized setting. Fixed-parameter algorithms, parameterized by the number of nodes $k$ of the pattern graph, are known for example for paths~\cite{monien1985find,alon1995color,williams2009finding,sieves2017}, trees~\cite{alon1995color}, even cycles~\cite{yuster1997finding}, odd cycles~\cite{alon1995color}, and patterns of constant treewidth~\cite{alon1995color}. By contrast, $k$-clique detection is known to be W[1]-hard, suggesting that it does not have a fixed-parameter algorithm~\cite{DOWNEY1995109}.

Induced subgraph detection, on the other hand, is W[1]-hard even for paths of length $k$~\cite{chen2007parameterized}. Any induced or non-induced subgraph on $k$ nodes can be detected in $n^{\omega k/3 + O(1)}$ time, where $\omega < 2.3729$ is the matrix multiplication exponent, due to a classical result of Ne{\v{s}}et{\v{r}}il and Poljak~\cite{nevsetvril1985complexity}.

\paragraph{Distributed subgraph detection.} As mentioned above, distributed subgraph detection has also received attention in the distributed setting recently. In \congest, non-trivial upper bounds are known e.g. for path and tree detection~\cite{korhonen2017deterministic,disc2017property}, cycle detection~\cite{fraigniaud2017spaa,korhonen2017deterministic, fischer2018possibilities} and clique detection~\cite{censor2021tight}. Likewise, lower bounds have been studied for cycle detection~\cite{drucker13,korhonen2017deterministic} and cliques~\cite{czumaj2019detecting}, and pattern graphs requiring near-quadratic time are known to exist~\cite{fischer2018possibilities}. Triangle detection remains a particularly interesting open question––the best known upper bound is $n^{1/3+o(1)}$ rounds~\cite{chang2019improved}, but no lower bounds are known.

In the congested clique, triangles can be detected in $O(n^{0.158})$ rounds and odd $k$-cycles in $2^{O(k)} n^{0.158}$ rounds using fast matrix multiplication~\cite{censor2015algebraic}. Even cycles can be detected even faster, in $O(k^2)$ rounds for $k = O(\log n)$~\cite{partition-trees}. Moreover, any induced or non-induced subgraph detection for $k$-node patterns can be solved in $O(n^{1-2/k})$ rounds in congested clique~\cite{tritri}.

\paragraph{Distributed parameterized complexity.} Parameterized distributed algorithms have appeared implicitly in many of the above-mentioned subgraph detection works, and recently \citet{ben2018parameterized} and \citet{siebertz2019parameterized} have explicitly studied aspects of distributed parameterized complexity. In terms of structural parameters, \emph{maximum degree} is a standard parameter in distributed setting, and algorithms parameterized degeneracy has been studied for various problems and models~\cite{barenboim2018distributed,barenboim2010sublogarithmic,ghaffari2019distributed}. Recently, \citet{li2018distributed} has shown that the treewidth of the input graph can be approximated in $\tilde{O}(D)$ rounds in \congest, and many classical optimization problems that are fixed-parameter tractable w.r.t. treewidth can be solved in $\tilde{O}\bigl(\tw(G)^{O(\tw(G))} D\bigr)$ rounds in \congest, where $\tilde{O}$ hides polylogarithmic factors in $n$.


\section{Preliminaries}\label{sec:preliminaries}

\paragraph{Degeneracy.}  A graph $G$ is called $d$-\textit{degenerate} if every induced subgraph of $G$ has a vertex of degree at most $d$. The minimum number $d$ for which $G$ is $d$-\textit{generate} is called \textit{degeneracy} of $G$, denoted by $\degen(G)$. It is easy to see that every $d$-degenerate graph admits an acyclic orientation such that the out-degree of each vertex is at most $d$.

\paragraph{Vertex cover number.} A \textit{vertex cover} of $G$ is a subset of vertices $S \subseteq V(G)$ such that every edge in $E(G)$ is incident with at least one vertex in $S$. The \textit{vertex cover number} $\vc(G)$ of $G$ is the minimum size of a vertex cover of $G$.

\paragraph{Treewidth.} A \emph{tree decomposition} of a graph $G = (V, E)$ is a pair $(\EuScript{X}, T)$, where $\EuScript{X} = \{ X_1, X_2, \dotsc, X_m \}$ is a collection of subsets of $V$ and $T$ is a tree on $\{ 1, 2, \dotsc, m\}$, such that
\begin{enumerate}
 \item $\bigcup_{i = 1}^{m} X_i = V$,
 \item for all edges $e \in E$ there exist $i$ with $e \subseteq X_i$
 \item for all $i$, $j$ and $k$, if $j$ is on the (unique) path from $i$ to $k$ in $T$, then $X_i\cap X_k\subseteq X_j$.
\end{enumerate}
The \emph{width} of a tree-decomposition $(\EuScript{X}, T)$ is defined as $\max_{i} |X_i| -1$. The \emph{treewidth} of a graph $G$ is the minimum width over all possible tree decompositions of $G$. Connected graphs of treewidth $1$ are trees, and connected graphs of treewidth $2$ are \emph{series-parallel graphs} (see e.g.~\cite{graphclasses}.)

\paragraph{Lower bound families.} For unconditional lower bounds in the \congest model, we use the standard framework of reducing from two-party communication complexity. Let $f \colon \{0,1\}^{2k} \to \{0,1\}$ be a Boolean function. In~the two-party communication game on $f$, there are two players who receive a private $k$-bit string $x_0$ and $x_1$ as input, and the task is to have at least one of the players compute $f(x) = f(x_0, x_1)$.

The template for these reductions is captured by \emph{families of lower bound graphs}:

\begin{definition}[e.g.~\cite{drucker13,frischknecht2012,AbboudCK16}]\label{def:lower-bound-family}
    Let $f_n \colon \{ 0, 1 \}^{2k(n)} \to \{ 0, 1 \}$ and $C \colon \mathbb{N} \to \mathbb{N}$ be functions and $\Pi$ a graph predicate. Suppose there is $n_0$ such that for any $n \ge n_0$ and all $x_0,x_1 \in \{0,1\}^{k(n)}$ there exists a (weighted) graph $G(n,x_0,x_1)$ satisfying the following properties:
\begin{enumerate}
    \item $G(n, x_0,x_1)$ satisfies $\Pi$ if and only if $f_n(x_0,x_1)=1$,
    \item $G(n, x_0,x_1) = (V_0 \cup V_1, E_0 \cup E_1 \cup S)$, where
        \begin{enumerate}[label=(\alph*),noitemsep]
            \item $V_0$ and $V_1$ are disjoint and $|V_0 \cup V_1| = n$,
            \item $E_i \subseteq V_i \times V_i$ for $i \in \{0,1\}$,
            \item $S \subseteq V_0 \times V_1$ is a cut and has size at least $C(n)$, and
            \item subgraph $G_i = (V_i, E_i)$ only depends on $i$, $n$ and $x_i$, i.e., $G_i = G_i(n, x_i)$.
        \end{enumerate}
\end{enumerate}
We then say that  $\F = (\G(n))_{n \in I}$ is a \emph{family of lower bound graphs}, where \[\G(n) = \{ G(n,x_0, x_1) \colon x_0, x_1 \in \{0,1\}^{k(n)} \}\,.\]
\end{definition}

\emph{Deterministic communication complexity} $\cc(f)$ of a function $f$ is the maximum number of bits the two players need to exchange in the worst case, over all deterministic protocols and input strings, in order to compute $f(x_0,x_1)$. \emph{Randomized communication complexity} $\rcc(f)$ is the worst-case complexity of protocols which compute $f$ with probability at least $2/3$ on all inputs.

\begin{theorem}[e.g.~\cite{drucker13,frischknecht2012,AbboudCK16}]\label{lower-bounds}
    Let $\F$ be a family of lower bound graphs. Any algorithm deciding $\Pi$ on a graph family $\EuScript{H}$ containing $\bigcup \G(n)$ for all $n \ge n_0$ in the \congest model with bandwidth $b(n)$ needs
    \[
    \Omega\left( \frac{\cc(f_n)}{C(n) b(n)} \right) \quad \textrm{ and } \quad \Omega\left( \frac{\rcc(f_n)}{C(n) b(n)} \right)
    \]
    deterministic and randomized rounds, respectively.
\end{theorem}

We reduce from the two-player \emph{set disjointness} function $\mathsf{DISJ}_{n} \colon \{ 0, 1 \}^{2n} \to \{ 0, 1 \}$, defined as $\mathsf{DISJ}_{n}(x_0,x_1) = 0$ if and only there is $i \in [n]$ such that $x_0(i) = x_1(i) = 1$. The communication complexity of set disjointness is $\cc(\mathsf{DISJ}_n) = \Omega(n)$  and $\rcc(\mathsf{DISJ}_n) = \Omega(n)$~\cite{KushilevitzN97,Razborov92}.

\section{Induced subgraph detection on general graphs}

\subsection{Patterns with cliques and independent sets: framework}

For the complexity results on detecting pattern graphs that contain a large independent set or a clique, we borrow the centralized reduction of Dalirrooyfard et al.~\cite{10.1145/3313276.3316329}. We present the reduction here in full, as we will need to analyze its implementation in the distributed setting.

We will start from instance $G$ of $s$-clique detection. The reduction will transform $G$ into an instance of (induced) $H$-detection, where the pattern graph $H$ contains a clique of size $s$, while increasing the number of nodes by a small factor. We first need the following definition:

\begin{definition}[\cite{10.1145/3313276.3316329}]
Let $G = (V,E)$ be a graph. A family $\C \subseteq 2^V$ is an \emph{$s$-clique cover}~if
\begin{enumerate}
    \item for each $s$-clique $K$ in $G$, there is a $C \in \C$ that contains the nodes of $K$, and
    \item the induced subgraph $G[C]$ is $s$-colorable for each $C \in \C$.
\end{enumerate}
We say that $\C$ is a \emph{minimum $s$-clique cover} if all $s$-clique covers of $G$ have at least $|\C|$ sets.
\end{definition}

Note that if $\C$ is a minimum $s$-clique cover, all induced subgraphs $G[C]$ for $C \in \C$ contain an $s$-clique, and thus require exactly $s$ colors to color.

\paragraph{Reduction overview.} Let $G = (V_G, E_G)$ be the original graph and let $H = ( V_H, E_H)$ be the pattern graph. Let $\C = \{ C_1, C_2, \dotsc, C_t \}$ be a minimum $s$-clique cover of $H$. We construct a graph $G^*$ as from the input graph $G$ follows:

\begin{enumerate}

\item The node set $V_{G^*}$ of $G^*$ consists of the following nodes:
\begin{enumerate}
    \item For each $i \in C_1$, there is a copy $V_{G,i} = V_G \times \{ i \}$ of the node set of $G$.
    \item For each $j \in V_H \setminus C_{1}$, there is a copy $j^*$ of the node $j$ in $G^{*}$.
\end{enumerate}
\item The edge set of $G^*$ is defined by the following rules:
\begin{enumerate}
    \item Each $V_{G,i}$ is an independent set.
    \item For each $i,j \in C_1$ and $v, u \in V_G$, we add edge between $(v,i)$ and $(u,j)$ if both $\{ i, j \} \in E_H$ and $\{ v, u \} \in E_G$.
    \item For each $i \in C_1$ and $j \in V_H \setminus C_1$ with $\{ i, j \} \in E_H$, we add edges between $j^*$ and all nodes $(v,i)$ for $v \in V_G$.
    \item For each $i, j \in V_H \setminus C_1$ with $\{ i, j \} \in E_H$, we add edge between $i^*$ and $j^*$.
\end{enumerate}
\end{enumerate}
Note that the graph $G^{*}$ has  $sn + |V_H|$ nodes.

\begin{lemma}[\cite{10.1145/3313276.3316329}]
If $G$ has an $s$-clique, then $G^*$ has $H$ as an induced subgraph, and if $G^*$ has $H$ as a subgraph, then $G$ has an $s$-clique.
\end{lemma}

\begin{proof}[Proof sketch.]
First assume that $G$ has an $s$-clique $K = \{v_{1}, \dots, v_{s} \}$. Let $\chi \colon C_1 \to [s]$ be a valid $s$-coloring of $H[C_1]$. We define a mapping $h$ from $V_H$ to the nodes of $G^*$ by
\[
h(i) =
\begin{cases}
    (v_{\chi(i)},i) & \text{ if $i \in C_1$, and}\\
    i^*             & \text{ if $i \notin C_1$.}
\end{cases}
\]
Now $G^*[h(V_H)]$ is isomorphic to $H$ by the construction of $G^*$ and the assumption that $K$ is a clique.

For the second direction, assume that $G^*$ subgraph $H^*$ isomorphic to $H$. First, assume that $H^*$ has an $s$-clique $K$ contained in $\bigcup_{i \in C_1} V_{G,i}$. Consider any two nodes $(v,i)$ and $(u,j)$ in $K$; since they are neighbors in $G^*$, we must have by construction of $G^*$ that $v \ne u$ and $i \ne j$, and we have $\{ v, u \} \in E_H$. Thus, the set of nodes in $V_G$ corresponding to $K$ in $G^*$, that is, $\{ v \in V_G \colon (v,i) \in K \}$, is an $s$-clique in $G$.

To finish the proof, we show that $H^*$ must have an $s$-clique $K$ contained in $\bigcup_{i \in C_1} V_{G,i}$. The proof is by contradiction; if this is not the case, we can construct an $s$-clique cover of $H$ that has less that $t = |\C|$ sets, which contradicts the minimality of $\C$. Thus, assume any $s$-clique in $H^*$ contains at least one node $j^*$ for some $j \in V_H$.

Now define a function $g \colon V_{G^*} \to V_H$ as
\[
h(w) =
\begin{cases}
    i   & \text{ if $w \in V_{G_i}$ for some $i \in C_1$, and}\\
    j^* & \text{ if $w = j^*$ for some $j \in V_H\setminus C_1$,}
\end{cases}
\]
that is, all nodes of $G^*$ are mapped to the `corresponding' nodes of $H$ as per construction of $G^*$; note that $h$ does not necessarily agree with the isomorphism between $H^*$ and $H$. Denoting by $V_{H^*}$ the node set of $H^*$, let $D_i = h^{-1}(C_i) \cap V_{H^*}$. We now show that the family $\D = \{ D_2, D_3, \dotsc, D_t \}$ is an $s$-clique cover of $H^*$ and $H^*$ is isomorphic to $H$, this implies that $H$ has an $s$-clique cover of size $t-1$. We observe that $\D$ satisfies the definition of $s$-clique cover:
\begin{itemize}
    \item \emph{Each $s$-clique of $H^*$ is contained in some $D_i \in \D$}: Consider an $s$-clique $K = \{ w_1, w_2, \dotsc, w_s \}$  in $H^*$. Let $L = h(K) \subseteq V_H$ be the image of $K$ under $h$. We have that $L$ is an $s$-clique in $H$: since $K$ is a clique in $G^*$, it can contain only one node from each independent set $V_{G,i}$, so all nodes in $K$ map to different nodes in $H$, and furthermore all nodes in $L$ are adjacent by the construction of $G^*$. 
    
    Since by assumption, $K$ contains at least one node $j^*$ for some $j \in V_H \setminus C_1$, we have that $f(j^*) = j \in L$ is not in $C_1$. Thus, the clique $L$ in $H$ is not covered by $C_1$, and there is $C_i \in \C \setminus \{ C_1 \}$ that covers $L$. It follows that the clique $K$ in $H^*$ is covered by $D_i$.
    \item \emph{Induced subgraph $H^*[D_i]$ is $s$-colorable for every $D_i \in \D$}: By assumption, there is an $s$-coloring $\chi$ of $H[C_i]$. We obtain a proper $s$-coloring of $H^*[D_i]$ by defining the coloring as $\phi(w) = \chi(h(v))$.
\end{itemize}
\end{proof}

\subsection{Patterns with cliques and independent sets: implications}\label{sec:induced-path-hardness}

\paragraph{Implementing the reduction in the congested clique.} Let $H$ be a pattern graph on $k$ nodes containing an $s$-clique. We now show that the reduction we gave above can be implemented efficiently in the congested clique model.

Assume we have algorithm $\A$ for (induced) $H$-detection running in $O(n^\delta)$ rounds in the congested clique. We now show that we can implement the above reduction in the congested clique to obtain an algorithm for detecting an $s$-clique, as follows:
\begin{enumerate}
    \item Each node $v \in V_G$ simulates nodes $(v, i)$ for $i \in C_1$, as well as one node from $V_H$.
    \item Since the incident edges of $(v,i)$ for $i \in C_1$ and nodes in $V_H \setminus C_1$ in $G^*$ only depend on the pattern graph $H$ and on the edges incident to $v$ in $G$, node $v$ can construct the inputs of its simulated nodes locally.
    \item Nodes then simulate the execution of $\A$ on a congested clique with $O(sn + k) = O(kn)$ nodes. The running time of $\A$ on the simulated instance is $O\bigl((kn)^\delta\bigr)$, and the simulation incurs additional overhead of $O(k^2)$, for a total running time of $O(k^{2\delta} n^\delta)$.
\end{enumerate}

Thus, we obtain the following:

\begin{theorem}\label{thm:path-lower-bound}
Let $H$ be a pattern graph with $k$ nodes that has a clique of size $s$. Then if we can solve $H$-detection or induced $H$-detection in the congested clique model in $O(n^\delta)$ rounds, we can find an $s$-clique in the congested clique in $O(k^{2\delta} n^\delta)$ rounds.
\end{theorem}

As an immediate corollary, we obtain a similar hardness result for induced subgraph detection for pattern graphs with large independent set, by observing that we can simply complement the pattern and input graphs. Note that this version only applies for \emph{induced} subgraph detection.

\begin{corollary}\label{cor:is-lower-bound}
Let $H$ be a pattern graph with $k$ nodes that has an independent set of size $s$. Then if we can solve induced $H$-detection in the congested clique model in $O(n^\delta)$ rounds, we can find an $s$-clique in the congested clique in $O(k^{2\delta} n^\delta)$ rounds.
\end{corollary}

\paragraph{Induced path detection.} Corollary~\ref{cor:is-lower-bound} immediately implies a conditional lower bound for induced path detection in the \congest model, as paths contain large independent sets:

\begin{corollary}
Let $k$ be fixed. If an induced $2k$-edge path or an induced $(2k+1)$-edge path can be detected in $O(n^\delta)$ rounds in the \congest model, then a $k$-clique can be detected in $O(n^\delta)$ rounds in the congested clique model. In particular, if an induced $4$-edge path can be detected in $O(n^\delta)$ rounds in the \congest model, then triangles can be detected $O(n^\delta)$ rounds in the congested clique model.
\end{corollary}

\paragraph{Patterns with cliques in \congest}\label{sec:lb-pattern-with-clique}

As a further application of the reduction of~\citet{10.1145/3313276.3316329}, we can transform the unconditional lower bound of \citet{czumaj2019detecting} for $4$-clique detection in \congest into a lower bound for induced subgraph detection for any pattern containing a $4$-clique.

\begin{lemma}[\cite{czumaj2019detecting}]\label{lemma:clique-lb-family}
Let $\Pi$ the graph predicate for existence of a $4$-clique. There exists a family of lower bound graphs for $\Pi$ with $f_n = \disj_{\Theta(n^2)}$ and $C(n) = \Theta(n^{3/2})$.
\end{lemma}

\begin{lemma}\label{lemma:lb-with-4clique}
Let $H$ be a pattern graph on $k$ nodes that contains a $4$-clique, and let $\Pi$ the graph predicate for existence of either induced or non-induced copy of $H$. Then there exists a family of lower bound graphs for $\Pi$ with $f_n = \disj_{\Theta(n^2)}$ and $C(n) = \Theta(n^{3/2})$.
\end{lemma}

\begin{proof}
Let $\F$ be the family of lower bound graphs for the existence of $4$-cliques given by Lemma~\ref{lemma:clique-lb-family}. We now want to apply the reduction given in Section~\ref{sec:induced-path-hardness} to construct a lower bound family $\F^*$ for $\Pi$. That is, we take $F^*$ to be the family of graphs $G^*(N,x_0,x_1)$, where $G^*(N,x_0,x_1)$ is obtained by applying the reduction to $G(n,x_0,x_1)$ and $N = sn + (k-s)$. 

First, observe that when the reduction with pattern graph $H$ is applied to a graph $G$, the resulting graph $G^*$ has $4n + k$ nodes. Moreover, since $\F$ is a family of lower bound graphs, it follows that $\F^*$ is satisfies condition (1) of Definition~\ref{def:lower-bound-family} with $f_n = \disj_{\Theta(n^2)}$.

Thus, it remains to give a partition of the node set of $G^*$ that satisfies condition (2) of Definition~\ref{def:lower-bound-family}. Denoting by $V_0$ and $V_1$ the partition in the original graph $G = G(n,x_0,x_1)$, we define the partition into $V^*_0$ and $V^*_1$ as follows:
\begin{enumerate}
    \item For each independent set $V_{G,j} = V \times \{ j \}$ in $G^*$ corresponding to a node $j \in C_1$ in $H$, we include the nodes $V_0 \times \{ j \}$ in $V^*_0$, and the nodes $V_1 \times \{ j \}$ in $V^*_1$.
    \item All nodes that are copies $j^*$ of nodes $j \in V_H \setminus C_1$ are included in $V_0$.
\end{enumerate}
Now since all the edges that depend on $x_i$ in $G(n,x_0,x_1)$ are in $E_i \subseteq V_i \times V_i$, all edges in $G^*$ that depend on $x_i$ are between nodes in $V_i \times C_1 \subseteq V_i^*$. Moreover, the edges crossing the cut between $V^*_0$ and $V_1^*$ are (a) the copies of cut edges $S$ in $G$ between two independent sets $V_{G,j}$ and $V_{G,k}$, and (b) then edges between the independent sets and the copies of nodes in $V_H \setminus C_1$. Thus, the total number of cut edges is at most $8C(n) + (k-4)n$, which is $O(n^{3/2})$ for constant $k$.
\end{proof}

Theorem~\ref{lower-bounds} and Lemma~\ref{lemma:lb-with-4clique} now immediately imply the following: 

\begin{theorem}
Let $H$ be a pattern graph that contains a $4$-clique. Any \congest algorithm solving either $H$-detection or induced $H$-detection needs at least $\Omega(n^{1/2}/\log n)$ rounds.
\end{theorem}

\subsection{Induced even cycle detection}\label{sec:induced-even-cycle-lb}

We next prove an unconditional lower bound for induced even cycle detection in \congest. Note that for induced odd cycles, one can easily verify that the construction of \citet{drucker13} immediately implies a $\Omega(n/\log n)$ lower bound.

\begin{lemma}\label{lemma:even-cycle-lb}
Let $k \ge 3$ be fixed, and let $\Pi$ the graph predicate for existence of an induced $2k$-cycle. There exists a family of lower bound graphs for $\Pi$ with $f_n = \disj_{\Theta(n^2)}$ and $C(n) = n$.
\end{lemma}

\paragraph{Construction.} For simplicity, let $k = 3$, and let $N \in \mathbb{N}$ and $x_0, x_1 \in \{ 0, 1 \}^{N^2}$.  For $n = 6N$, we construct a graph $G(n,x_0,x_1)$ as follows:
\begin{enumerate}
    \item The node set of $G(n,x_0,x_1)$ contains four sets of $N$ nodes, denoted by $A_1$, $A_2$, $B_1$, and $B_2$, as well as further $2N$ nodes as described below. For the purposes of the construction, we denote the nodes in these sets as $A_i = \{ a_{i,1}, a_{i,2}, \dotsc, a_{i,N}\}$ and $B_i = \{ b_{i,1}, b_{i,2}, \dotsc, b_{i,N}\}$.
    \item Graph $G(n,x_0,x_1)$ has the following edges for all $x_0$ and $x_1$:
    \begin{enumerate}[label=(\alph*),noitemsep]
        \item Each set $A_1$, $A_2$, $B_1$, and $B_2$ is a clique.
        \item For $i = 1, 2, \dotsc, N$, there is a path of length $3$ between node $a_{1,i}$ and node $b_{1,i}$.
        \item For $i = 1, 2, \dotsc, N$, there is an edge between node $a_{2,i}$ and node $b_{2,i}$.
    \end{enumerate}
    \item Depending on $x_0$ and $x_1$, we add the following edges to the graph. Fix an arbitrary bijection between $[N^2]$ and $[N] \times [N]$. For each $i \in [N^2]$ such that $x_0(i) = 1$, we add the edge $\{ a_{1,j}, a_{2,k} \}$ for the corresponding pair $(j,k)$. Likewise, for each $i \in [N^2]$ such that $x_1(i) = 1$, we add the edge $\{ b_{1,j}, b_{2,k} \}$ for the corresponding pair $(j,k)$.
\end{enumerate}
The partition required by Definition~\ref{def:lower-bound-family} is obtained by setting $V_0 = A_1 \cup A_2$. The size of the cut $S$ is $2N$.

\begin{lemma}
Let $G(n, x_0, x_1)$ be the graph defined above. Then we have that $G(n, x_0, x_1)$ contains an induced $6$-cycle if and only if $x_0$ and $x_1$ are not disjoint.
\end{lemma}

\begin{proof}
First, we observe that if there is $i \in [N^2]$ such that $x_0(i) = x_1(i) = 1$, then $G(n, x_0, x_1)$ has an induced $6$-cycle of from
\[ \bigr(a_{1,j}, b_{1,j}, b_{2,k}, p_1, p_2, a_{2,k}, a_{1,j} \bigl)\,,\]
where $(j,k) \in [N] \times [N]$ is the pair corresponding to $i$, and $p_1$ and $p_2$ are the nodes on the path between $a_{2,k}$ and $b_{2,k}$. We now claim that if there is an induced $6$-cycle in $G(n, x_0, x_1)$, it must have this form, implying that $x_0$ and $x_1$ are not disjoint.

Assume that $C$ is an induced $6$-cycle in $G(n, x_0, x_1)$, and let $X$ be one of the sets $A_1$, $A_2$, $B_1$, and $B_2$. Let us call the edges between $A_1$ and $A_2$, and edges between $B_1$ and $B_2$ \emph{inward} edges, and the remaining edges \emph{outward} edges. We will first show that if $C$ contains any nodes from $X$, then $C$ is has one of the following \emph{types} with regard to $X$:
\begin{enumerate}[label=T\arabic*,noitemsep]
    \item There is one node $v$ from $X$ in $C$, and $v$ is incident to one inward edge and one outward edge in $C$.
    \item There is one node $v$ from $X$ in $C$, and $v$ is incident to two inward edges in $C$.
    \item There are two nodes $u$, $v$ from $X$ in $C$, and the nodes $u$ and $v$ are incident in $C$.
    \begin{enumerate}[label*=(\alph*),noitemsep]
        \item Node $u$ is incident to an outward edge and $v$ is incident to an inward edge.
        \item Both $u$ and $v$ are incident to a single outward edge in $C$.
        \item Both $u$ and $v$ are incident to a single inward edge in $C$.
    \end{enumerate}
\end{enumerate}
\begin{center}
\includegraphics[width=0.75\textwidth]{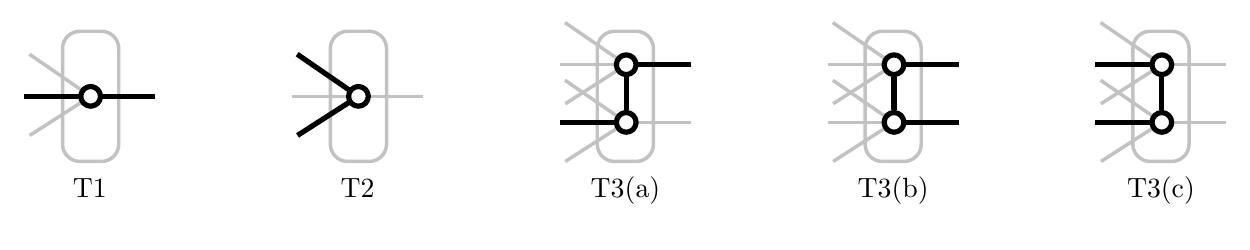}
\end{center}

To see why this is the case, we first observe that $C$ cannot contain three or more nodes from $X$, since it would then contain a chord. Each node in $X$ is incident to exactly one outward edge, so if $C$ contains a single $v$ node from $X$, then it must be of type T1 or T2. The remaining case is that $C$ contains two nodes from $X$; these must be incident on $C$, since otherwise $C$ would have a chord. 

By the construction of $G(n, x_0, x_1)$, the induced cycle $C$ must contain a node from one of $A_1$, $A_2$, $B_1$, and $B_2$; by symmetry, we can assume without loss of generality that $C$ contains a node from either $A_1$ or $A_2$. We show that $C$ has to be type T1 w.r.t. both $A_1$ and $A_2$, as in all other cases we get a contradiction:
\begin{enumerate}[label=\emph{Case \arabic*:}]
    \item Cycle $C$ is type T2 w.r.t.~$A_1$. Then $C$ is type T3(c) w.r.t.~$A_2$, which implies that $C$ is triangle. (\emph{Symmetric case: $C$ is type T2 w.r.t. $A_2$.})
    \item Cycle $C$ is type T3(c) w.r.t.~$A_1$. Then $C$ is either type T2 or T3(c) w.r.t.~$A_2$, which implies that $C$ is a triangle or a $4$-cycle. (\emph{Symmetric case: $C$ is type T3(c) w.r.t.~$A_2$.})
    \item Cycle $C$ is type T3(b) w.r.t.~$A_1$. Then $C$ is type T3(b) w.r.t.~$B_1$, which implies that $C$ is a $4$-cycle.
    \item Cycle $C$ is type T3(b) w.r.t.~$A_2$. Then $C$ is type T3(b) w.r.t.~$B_2$, which implies that $C$ is an $8$-cycle.
    \item Cycle $C$ is type T3(a) w.r.t.~$A_1$. Then $C$ is either type T1 or T3(a) w.r.t.~$A_2$, $B_2$, and $B_1$. This implies that $C$ must have length at least $7$. (\emph{Symmetric case: $C$ is type T3(a) w.r.t.~$A_2$.})
\end{enumerate}
\begin{center}
\includegraphics[width=0.95\textwidth]{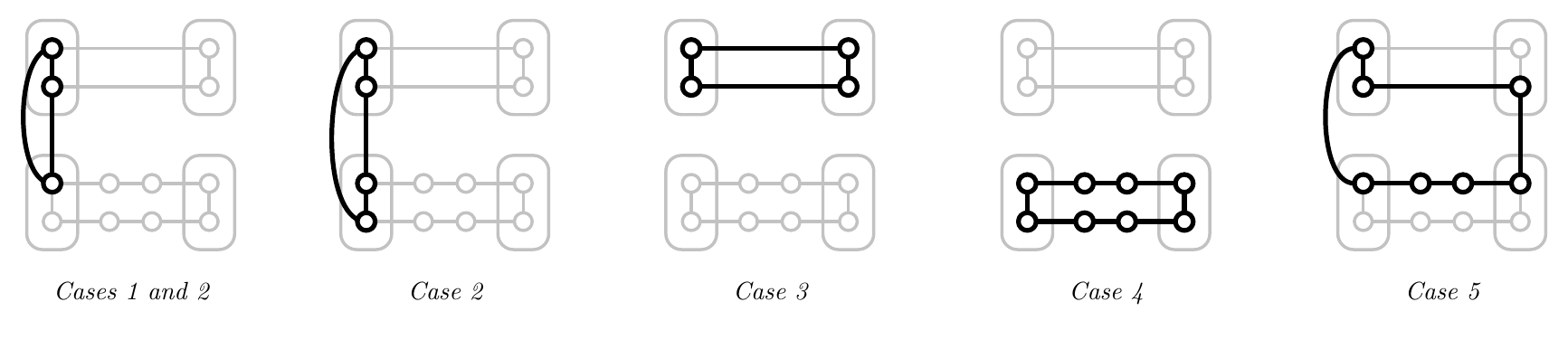}
\end{center}
Thus, the only case that remains is that $C$ is type T1 w.r.t.~$A_1$, $A_2$, $B_2$, and $B_1$. In this case, $C$ contains $a_{1,j}$ and $b_{1,j}$ for some $i \in [N]$, as well as $a_{2,k}$ and $b_{2,k}$ for some $i \in [N]$, implying the edges $\{ a_{1,j}, a_{2,k} \}$ and $\{ b_{1,j}, b_{2,k} \}$ exist in $G(n, x_0, x_1)$.
\end{proof}

Theorem~\ref{lower-bounds} and Lemma~\ref{lemma:even-cycle-lb} now immediately imply the following: 

\begin{theorem}
Any \congest algorithm solving induced $2k$-cycle detection for $k \ge 3$ needs at least $\Omega(n/\log n)$ rounds.
\end{theorem}

\subsection{Induced subgraph detection for bounded treewidth patterns}\label{sec:lb-low-treewidth}

Finally, we consider subgraph and induced subgraph detection for pattern graphs of low treewidth. Recall that in centralized setting, a subgraph $H$ with treewidth $t$ can be detected in time $2^{O(k)} n^{t+1}\log n$~\cite{alon1995color}, implying that detecting constant-treewidth subgraphs is fixed-parameter tractable. However, in \congest model, turns out that pattern of treewidth $2$ are already maximally hard.

Our construction for the hard pattern graph uses similar ideas as the hard non-induced subgraph detection instances presented by \citet{fischer2018possibilities}. However, the pattern graphs they use a fairly dense and have treewidth higher than $2$.

\begin{theorem}\label{thm:bounded-treewidth-pattern}
For any $k \ge 2$, there exists a pattern graph $H_k$ of treewidth $2$ such that \congest algorithm solving either $H_k$-detection or induced $H_k$-detection needs at least $\Omega(n^{2-1/k})$ rounds.
\end{theorem}

Let $k \ge 2$ be fixed. We construct the graph $H_k$ as follows:
\begin{enumerate}
    \item We start with four triangles $A_1$, $A_2$, $B_1$ and $B_2$ with nodes labelled by $1$, $2$ and $3$.
    \item Nodes $1$ of $A_1$ and $A_2$ are connected by an edge, and nodes $1$ of $B_1$ and $B_2$ are connected by an edge.
    \item Nodes $2$ of $A_1$ and $B_1$ are connected with $k$ disjoint paths of length $3$. Likewise, Nodes $2$ of $A_2$ and $B_2$ are connected with $k$ disjoint paths of length $3$.
\end{enumerate}
The graph $H_k$ is a \emph{series-parallel} graph, and thus has treewidth~2~\cite{graphclasses}.

\begin{center}
\includegraphics[width=0.6\textwidth]{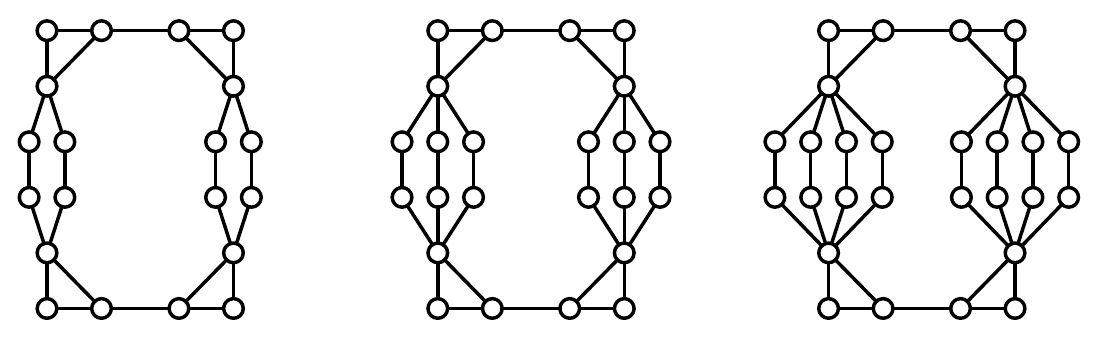}
\end{center}

\begin{lemma}\label{lemma:bounded-treewidth-lb}
Let $k \ge 2$ be fixed. There exists a family of lower bound graphs for $H_k$-detection and induced $H_k$-detection with $f_n = \disj_{\Theta(n^2)}$ and $C(n) = \Theta(n^{1/k})$.
\end{lemma}

\paragraph{Construction.} Let $N \in \mathbb{N}$ and $x_0, x_1 \in \{ 0, 1 \}^{N^2}$. We construct a graph $G(n,x_0,x_1)$ as follows:
\begin{enumerate}
    \item We start with $4N$ triangles $A_{1,1}, \dotsc, A_{1,N}$, $A_{2,1}, \dotsc, A_{2,N}$, $B_{1,1}, \dotsc, B_{1,N}$ and $B_{2,1}, \dotsc, B_{2,N}$. The nodes of each triangle are labelled with $1$, $2$ and $3$ for the purposes of the construction.
    \item Let $K$ be the smallest integer such that $\binom{n}{k} \ge K$; one can verify that $K \le \lceil k n^{1/k} \rceil$. We add $4$ sets of $K$ nodes, denote by $a_{i,j}$ and $b_{i,j}$ for $i \in \{ 1, 2\}$ and $j \in \{ 1, 2, \dotsc, K \}$.
    \item Graph $G(n,x_0,x_1)$ has the following additional edges for all $x_0$ and $x_1$:
    \begin{enumerate}[label=(\alph*),noitemsep]
        \item For each $j \in \{ 1, 2, \dotsc, K \}$, there are edges $\{ a_{1,j}, b_{1,j}\}$ and $\{ a_{2,j}, b_{2,j}\}$.
        \item Fix an injection $\rho$ from $[N]$ to subsets of $[K]$ of size $k$. For each $i \in \{ 1,2 \}$ and $j \in [N]$, we add edge between node $2$ of triangle $A_{i,j}$ and $a_{i,\ell}$ for each element $\ell \in \rho(j)$, and likewise between between node $2$ of triangle $B_{i,j}$ and $b_{i,\ell}$.
    \end{enumerate}
    \item Depending on $x_0$ and $x_1$, we add the following edges to the graph. Fix an arbitrary bijection between $[N^2]$ and $[N] \times [N]$. For each $i \in [N^2]$ such that $x_0(i) = 1$, we add the edge between node $1$ of triangle $A_{1,j}$ and node $1$ of triangle $A_{2,\ell}$ for the corresponding pair $(j,\ell)$. Likewise, for each $i \in [N^2]$ such that $x_1(i) = 1$, we add the edge between nodes $1$ of triangles $B_{1,j}$ and $B_{2,\ell}$ for the corresponding pair $(j,\ell)$.
\end{enumerate}
Figure~\ref{fig:bounded-treewidth-lb} shows an example of this construction.

The total number of nodes in the graph $G(n,x_0,x_1)$ is $n = 12N + 4 kN^{1/k} = \Theta(N)$. The partition required by Definition~\ref{def:lower-bound-family} is obtained by setting $V_0$ to contain all triangles $A_{i,j}$ and nodes $a_{i,\ell}$. The size of the cut $S$ is $kN^{1/k} = \Theta(n^{1/k})$.

\begin{lemma}\label{lemma:bounded-treewidth-lb-aux1}
Let $G(n, x_0, x_1)$ be the graph defined above. Then we have that $G(n, x_0, x_1)$ contains a copy of $H_k$ if and only if $x_0$ and $x_1$ are not disjoint. Moreover, any copy of $H_k$ in $G(n, x_0, x_1)$ is an induced copy of $H_k$.
\end{lemma}

\begin{proof}
First, we observe that if there is $i \in [N^2]$ such that $x_0(i) = x_1(i) = 1$, then $G(n, x_0, x_1)$ has a copy of $H_k$. Let $(j,\ell) \in [N] \times [N]$ be the pair corresponding to $i$ under the bijection fixed in the construction. We obtain a copy of $H_k$ by taking the triangles $A_{1,j}$, $A_{2,\ell}$, $B_{1,j}$, and $B_{2,\ell}$. By assumption that $x_0(i) = x_1(i) = 1$, triangles $A_{1,j}$ and $A_{2,\ell}$ are connected by an edge, as are $B_{1,j}$ and $B_{2,\ell}$. Furthermore, by construction, there are $k$ disjoint paths of length $3$ between $A_{1,j}$ and $B_{1,j}$, as well as between $A_{2,\ell}$ and $B_{2,\ell}$.

Now assume that $G(n, x_0, x_1)$ has a copy of $H_k$. The copy of $H_k$ contains triangles $A'_1$ and $A'_2$ connected by an edge, and triangles $B'_1$ and $B'_2$ connected by an edge. We first focus on triangles $A'_1$ and $A'_2$; by construction of $G(n, x_0, x_1)$, any pair of two triangles connected by an edge must consist of triangles $A_{1,j}$ and $A_{2,\ell}$, or of triangles $B_{1,j}$ and $B_{2,\ell}$ for some $(j,\ell) \in [N] \times [N]$. By relabelling, we can thus assume that $A'_1 = A_{1,j}$ and $A'_2 = A_{2,\ell}$ for some $(j,\ell) \in [N] \times [N]$.

Now consider the triangle $B'_1$ in the copy of $H_k$. There are $k$ disjoint paths of length $3$ starting from node $2$ of $A'_1 = A_{1,j}$ and ending at some node of $B'_1$; this is only possible if $B'_1 = B_{1,j}$. Likewise, we have $B'_2 = B_{2,\ell}$. Since there is an edge between triangles $A_{1,j}$ and $A_{2,\ell}$ as well as triangles $B_{1,j}$ and $B_{2,\ell}$, we must have $x_0(i) = x_1(i) = 1$ for $i$ corresponding to pair $(j,\ell)$. Finally, we note that by construction of $G(n,x_0,x_1)$, the copy of $H_k$ is an induced copy.
\end{proof}

Theorem~\ref{thm:bounded-treewidth-pattern} now follows immediately by Theorem~\ref{lower-bounds}.

\begin{figure}
\begin{center}
\includegraphics[width=0.5\textwidth]{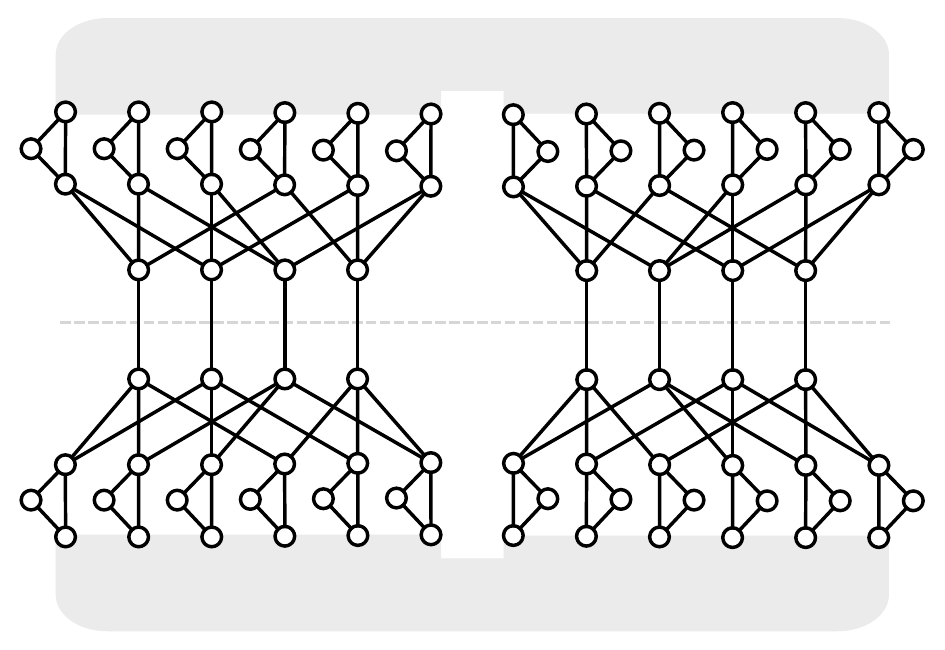}
\caption{Example of the hard instance for $H_k$-detection for $k = 2$. The shaded area is represents the encoding edges, which form a subgraph of a complete bipartite graph.}\label{fig:bounded-treewidth-lb}
\end{center}    
\end{figure}

\section{Multicolored problems}

In the \emph{multicolored (induced) subgraph detection}, we are given a pattern graph $H$ on $k$ nodes and an input graph $G$ with a (not necessarily proper) $k$-coloring, and the task is to find a (induced) copy of $H$ that is multicolored, i.e. a copy where all nodes have different colors.

\subsection{Reductions}\label{sec:multicolor-reductions}

We first prove that the complexities of multicolored $k$-clique and $k$-independent set are closely related to their standard versions also in the distributed setting. These results follow from standard fixed-parameter reductions~\cite{PIETRZAK2003757,FELLOWS200953}.

\begin{theorem}\label{thm:multicol-red}
If multicolored $k$-clique can be solved in $T(n)$ rounds in \congest, then $k$-clique can be solved $O(k^2 T(kn))$ rounds in \congest{}. If $k$-clique can be solved in $T(n)$ rounds in \congest, then multicolored $k$-clique can be solved $T(n)$ rounds in \congest{}.
\end{theorem}

\begin{proof}
First assume we have a $T(n)$-round algorithm $\A$ for multicolored $k$-clique. We can then solve $k$-clique by simulating $\A$ as follows:
\begin{enumerate}
    \item Each node $v \in V$ simulates $k$ new nodes $v_1, v_2, \dotsc, v_k$, with $v_i$ having color $i$; these nodes form an independent set. For each neighbor $u \in N(v)$, all nodes $v_i$ are adjacent to $u_j$ for all $i, j \in [k]$.
    \item Nodes simulate $\A$ on the modified graph. This takes $O(kn)$ rounds on the modified graph, and simulation gives $O(k^2)$ overhead in the round complexity.
\end{enumerate}
By construction, any multicolor $k$-clique in the modified graph can contain at most one node $v_i$ for any $v \in V$, and the multicolored $k$-clique corresponds to a $k$-clique in $G$.

For the second part, assume we have a $T(n)$-round algorithm $\A$ for $k$-clique. Given a colored graph $G$, we can solve the multicolored $k$-clique by removing all edges between nodes of the same color, and running algorithm $\A$ on the remaining graph.
\end{proof}

In the centralized setting, clique and independent set are equivalent, so the above reductions work also for independent set. However, in the distributed setting, only one direction works immediately.

\begin{theorem}
If multicolored $k$-independent set can be solved in $T(n)$ rounds in \congest, then $k$-independent set can be solved $O(k^2 T(kn))$ rounds in \congest{}.
\end{theorem}

\begin{proof}
The claim follows by an identical argument as in the previous theorem, with the exception that the simulated nodes $v_1, v_2, \dotsc, k$ form a clique instead of an independent set.
\end{proof}

\subsection{Lower bounds}\label{sec:multicolor-lb}

Next, we prove some simple unconditional lower bounds for multicolored (induced) cycle detection and multicolored induced path detection.

\begin{theorem}
For any $k \ge 4$, any \congest algorithm solving multicolored (induced) $k$-cycle detection needs at least $\Omega(n/\log n)$ rounds.
\end{theorem}

\begin{proof}
The claim follows by adding colors to the standard odd cycle lower bound family construction of \citet{drucker13}. In more detail, let $N \in \mathbb{N}$ and $x_0, x_1 \in \{ 0, 1 \}^{N^2}$.  For $n = (k+1) N$, we construct a lower bound  graph $G(n,x_0,x_1)$ as follows:
\begin{enumerate}
    \item The node set of $G(n,x_0,x_1)$ contains four sets of $N$ nodes, denoted by $A_1$, $A_2$, $B_1$, and $B_2$. For the purposes of the construction, we denote the nodes in these sets as $A_i = \{ a_{i,1}, a_{i,2}, \dotsc, a_{i,N}\}$ and $B_i = \{ b_{i,1}, b_{i,2}, \dotsc, b_{i,N}\}$.
    \item Graph $G(n,x_0,x_1)$ has the following edges for all $x_0$ and $x_1$:
    \begin{enumerate}[label=(\alph*),noitemsep]
        \item For $i = 1, 2, \dotsc, N$, there is a path of $k-3$ edges between node $a_{1,i}$ and node~$b_{1,i}$.
        \item For $i = 1, 2, \dotsc, N$, there is an edge between node $a_{2,i}$ and node $b_{2,i}$.
    \end{enumerate}
    \item Depending on $x_0$ and $x_1$, we add the following edges to the graph. Fix an arbitrary bijection between $[N^2]$ and $[N] \times [N]$. For each $i \in [N^2]$ such that $x_0(i) = 1$, we add the edge $\{ a_{1,j}, a_{2,\ell} \}$ for the corresponding pair $(j,\ell)$. Likewise, for each $i \in [N^2]$ such that $x_1(i) = 1$, we add the edge $\{ b_{1,j}, b_{2,\ell} \}$ for the corresponding pair $(j,\ell)$.
    \item The nodes in sets  $A_1$, $A_2$, $B_1$, and $B_2$ are colored with colors $1$, $2$, $3$ and $4$, respectively. Each path between $A_1$ and $B_1$ uses each of the remaining colors once.
\end{enumerate}
The partition required by Definition~\ref{def:lower-bound-family} is obtained by setting $V_0 = A_1 \cup A_2$. The size of the cut $S$ is $2N$.

To see that the graph $G(n,x_0,x_1)$ contains a multicolored $k$-cycle if and only if $x_0$ and $x_1$ intersect, we observe that any possible multicolored cycle must use edge $\{ a_{2,i}, b_{2,i} \}$ and a path between $a_{1,j}$ and $b_{1,j}$ for some $i,j \in [N]$. Thus, the edges $\{ a_{1,i}, a_{2,j} \}$ and $\{ b_{1,i}, b_{2,j} \}$ are present for some $i,j\in [N]$ if and only if $G(n,x_0,x_1)$ contains a multicolored $k$-cycle.
\end{proof}

\begin{theorem}
For any $k \ge 6$, any \congest algorithm solving multicolored induced $k$-edge path detection needs at least $\Omega(n/\log n)$ rounds.
\end{theorem}

\begin{proof}
For simplicity, we consider the case $k = 6$. Let $N \in \mathbb{N}$ and $x_0, x_1 \in \{ 0, 1 \}^{N^2}$.  For $n =  7N + 2$, we construct a lower bound graph $G(n,x_0,x_1)$ as follows:
\begin{enumerate}
    \item The node set of $G(n,x_0,x_1)$ contains five sets of $N$ nodes, denoted by $A_1$, $A_2$, $B_1$, $B_2$, and $C$. For the purposes of the construction, we denote the nodes in these sets as $A_i = \{ a_{i,1}, a_{i,2}, \dotsc, a_{i,N}\}$ and $B_i = \{ b_{i,1}, b_{i,2}, \dotsc, b_{i,N}\}$, and $C = \{ c_1, c_2 \dotsc, c_N \}$. In addition, we have two terminal nodes $s$ and $t$.
    \item For $i = 1, 2, \dotsc, N$, there is an edge between nodes $a_{2,i}$ and $b_{2,i}$, and between $c_i$ and~$b_{1,i}$.
    \item The nodes in $A_1$ and $C$ as well as the terminal nodes $s$ and $t$ form the gadget for controlling the end points of any multicolor $k$-edge path:
    \begin{enumerate}
        \item There is an edge between node $s$ and every node in $A_1$, and an edge between node $t$ and every node in $C$.
        \item There is a complete bipartite graph between $A_1$ and $C$, with every edge of form $\{ a_{1,i}, c_i \}$ removed.
    \end{enumerate} 
    \item The edges between $A_1$ and $A_2$ and the edges between $B_1$ and $B_2$ are used to encode the disjointness instance as in the previous proofs. That is, depending on $x_0$ and $x_1$, we add the following edges to the graph. Fix an arbitrary bijection between $[N^2]$ and $[N] \times [N]$. For each $i \in [N^2]$ such that $x_0(i) = 1$, we add the edge $\{ a_{1,j}, a_{2,\ell} \}$ for the corresponding pair $(j,\ell)$. Likewise, for each $i \in [N^2]$ such that $x_1(i) = 1$, we add the edge $\{ b_{1,j}, b_{2,\ell} \}$ for the corresponding pair $(j,\ell)$.
    \item The nodes in sets  $A_1$, $A_2$, $B_1$, $B_2$ and $C$ are colored with colors $2$, $3$, $4$, $5$ and $6$, respectively. The terminal nodes $s$ and $t$ have colors $1$ and $7$.
\end{enumerate}
The partition required by Definition~\ref{def:lower-bound-family} is obtained by setting $V_0 = A_1 \cup A_2 \cup C \cup \{ s, t \}$. The size of the cut $S$ is $2N$.

To see that the graph $G(n,x_0,x_1)$ satisfies the conditions of the lower bound family, consider any induced multicolored $k$-edge path $P$ in $G(n,x_0,x_1)$. Since all neighbors of $s$ have color $2$, node $s$ must be one endpoint of $P$; similarly, node $t$ is the other endpoint of $P$. Denote the node of $P$ in $A_1$ by $v_2$, and the node of $P$ in $C$ by $v_6$. Since the path $P$ has $k$ edges, nodes $v_2$ and $v_6$ cannot be adjacent on $P$, and since $P$ is induced, there cannot be an edge between $v_2$ and $v_2$ in $G(n,x_0,x_1)$. Thus, by construction, we have that $v_2 = a_{1,i}$ and $v_6 = c_i$ for some $i \in [N]$. Thus, the path $P$ must travel from $a_{1,i}$ to $c_i$ using one node of each remaining color, which is possible only if and only if $x_0$ and $x_1$ are not disjoint by the properties of the standard disjointness encoding.

Finally, one can observe that the proof works for $k \ge 7$ if one replaces the edges between $A_{2}$ and $B_2$ by paths of appropriate length.
\end{proof}

\section{Induced subgraph detection on bounded degeneracy graphs}\label{sec:induced-trees-degen}

\subsection{Induced tree detection}

We start by giving a parameterized distributed algorithm for detecting induced trees, parameterized by the degeneracy $\degen = \degen(G)$ of the input graph. This result is based on the \emph{random separation} algorithm of \citet{cai2006random}, adapted to distributed setting. For this result, we assume for convenience that all nodes are given the parameter $\degen$ as input; we discuss at the end how to remove this dependence for the randomized version of the algorithm.

\paragraph{Preliminaries.}

Let $G = (V,E)$ be a graph. We say that an orientation $\sigma$ of the edges of $G$ is an \emph{$\alpha$-bounded orientation}, or simply \emph{$\alpha$-orientation}, if every node $v \in V$ has out-degree at most $\alpha$ in $\sigma$, and $\sigma$ is acyclic. A graph $G$ is $\degen$-degenerate if and only if has an $\degen$-orientation; moreover, an $O(\degen)$-orientation can be computed fast in the \congest model:

\begin{lemma}[\cite{barenboim2010sublogarithmic}]\label{lemma:compute-orientation}
Let $G$ be a $\degen$-degenerate graph, and let $\varepsilon > 0$. We can compute a $(2+\varepsilon)\degen$-orientation of $G$ in $O(\log n)$ rounds in the \congest model, assuming $\degen$ is known to all nodes. If $\degen$ is not known, we instead can compute a $(4+\varepsilon)\degen$-orientation of $G$ in $O(\log n)$ rounds.
\end{lemma}

\paragraph{Multicolored induced trees with orientation.}

Let $T$ be a tree on $k$ nodes. We first to show how to solve a specific multicolored version of induced $T$-detection, given an acyclic orientation of $G$ as input. 

More precisely, let the graph $G$, let $\sigma$ be an $\alpha$-bounded orientation of $G$, and let $\chi \colon V \rightarrow \{ 0,1,\dots, k \}$ be a (not necessarily proper) $(k+1)$-coloring of $G$. Moreover, assume that the tree $T$ is labelled in a bottom-up manner with $1, 2, \dotsc, k$ with an arbitrary node as a root -- that is, the root has label $k$, and each node has a smaller label than their parent. We say that an induced copy $H$ of $T$ in $G$ is \emph{proper} w.r.t $\sigma$ and $\chi$ if the node in $H$ corresponding to node $i$ in $T$ has color $i$, and every node that is an out-neighbor of some node in $H$ has color $0$.

\begin{lemma}\label{lemma:proper-ind-subgraph}
Given a graph $G = (V,E)$, an orientation $\sigma$ of $G$, and a coloring $\chi$ as input, we can find a proper induced copy of a tree $T$ in $O(k)$ rounds using $O(1)$-bit messages in \congest model.
\end{lemma}

\begin{proof}
For each node $i$, let $T_{i}$ be the subtree of $T$ rooted at $i$ and $p(i)$ as the parent of $i$. The algorithm for detecting a proper induced copy of $T$ proceeds as follows:

\begin{enumerate}
\item The nodes broadcast their color to their neighbors; since there are $k+1$ colors, we can do this in $O(\log k)$ rounds using 1-bit messages. Each node now knows the colors of its neighbors.

\item For iterative steps $i = 1, 2, \dotsc, k - 1$, each node $v$ of color $i$ checks locally if the following conditions are satisfied:
\begin{enumerate}[label=(\alph*)]
    \item There is at most one out-neighbor of $v$ of color $p(i)$, and all other out-neighbors of $v$ have color $0$. 
    \item Node $v$ has received message $1$ from at least one node of color $j$, for each child $j$ of $i$ in $T$.
\end{enumerate}
If the conditions (a) and (b) are satisfied, node $v$ broadcasts $1$, and otherwise it broadcasts~$0$.

\item Finally, each node $v$ of color $k$ checks that all its out-neighbors have color $0$, and if the condition (b) above is satisfied for $i = k$. If these checks pass, $v$ reports an induced copy of $T$.
\end{enumerate}

For any proper induced copy of $T$, one can easily verify by induction on $i$ that node of color $i$ in that copy passes the checks performed by the algorithm. Moreover, any copy $T'$ reported by the algorithm must be a proper induced copy: if there was a non-tree edge from node $v$ of color $i$ to node $u$ of color $j$, then either $v$ or $u$ would have failed the corresponding check. The algorithm clearly uses $k + O(\log k)$ rounds of communication.

\end{proof}

\paragraph{Induced trees.} Using Lemma~\ref{lemma:proper-ind-subgraph} as a subroutine, we now show how to detect induced copies of any tree $T$. To achieve this, we use random separation~\cite{cai2006random} and color-coding~\cite{alon1995color} techniques to reduce the general problem to detection of proper induced copies of $T$.

\begin{theorem}\label{thm:path}
Finding induced copy of a tree $T$ on $k$ nodes in a $\degen$-degenerate graph $G$ can be done in $k2^{O(\degen k)}k^k + O(\log n)$ rounds in the \congest model using a randomized algorithm.
\end{theorem}

\begin{proof}

Given an input graph $G$ with degeneracy $\degen$, the algorithms is as follows:
\begin{enumerate}
    \item Compute an $\alpha$-orientation $\sigma$ of $G$ for $\alpha = 5\degen$. This can be done in $O(\log n)$ rounds by Lemma~\ref{lemma:compute-orientation}.
    \item Perform the following steps $t$ times, for $t$ to be determined later:
    \begin{enumerate}[label=(\alph*)]
        \item Each node $v$ picks color $\chi(v) = 0$ with probability $1/2$.
        \item If node $v$ did not color itself at previous step, they pick a color $\chi(v)$ from set $\{ 1, 2, \dotsc, k \}$ uniformly at random.
        \item The nodes use the algorithm of Lemma~\ref{lemma:proper-ind-subgraph} to detect if there is a proper induced copy of $T$ w.r.t.~$\sigma$ and $\chi$.
    \end{enumerate}
\end{enumerate}

If there is an induced copy $H$ of $T$ in $G$, then under the orientation $\sigma$, there are at most $5 \degen k$ nodes that are out-neighbors of nodes of $H$. Thus, the probability that the coloring process of steps (a) and (b) gives a coloring $\chi$ such that $H$ is proper is at least $2^{-5\degen k}k^{-k}$.
Thus, if we pick $t = \Theta(2^{5\degen k}k^k \log n)$, the algorithm detects an induced copy of $T$ with a high probability if one exists.

For the running time, computing the orientation takes $O(\log n)$ rounds. We repeat the color-coding and proper induced subgraph detection $t = \Theta(2^{5 \degen k}k^k \log n)$ times; by Lemma~\ref{lemma:proper-ind-subgraph}, this part can be implemented with $O(1)$-bit messages, so we can perform $\Omega(\log n)$ instances in parallel. Thus, the total round cost for this part is $O(k2^{5\degen k}k^k)$ rounds.
\end{proof}

\paragraph{Derandomization.}

Finally, we note that the algorithms can be derandomized using standard derandomization tools from fixed-parameter algorithms. Specifically, we use the derandomization of \citet{alon2007linear} to avoid incurring extra $O(\log n)$ factor that would follow from the original derandomization of \citet{cai2006random}. We first recall the definition of \emph{families of perfect hash functions}.

\begin{definition}[\cite{cygan2015parameterized}]
An \emph{$(n,k)$-family of perfect hash functions $\F$} is a family of functions $f \colon [n] \to [k]$ such that for any $S \subseteq [n]$ with $|S| = k$ there is a function $f \in \F$ that is bijection on $S$, i.e. $f(S) = [k]$.
\end{definition}

Families of perfect hash functions can be constructed efficiently in the centralized setting, as captured by the following lemma. However, as we are not considering the local computation cost of the algorithm, we are mainly interested in the size bound given the lemma.

\begin{lemma}[\cite{naor1995splitters}]\label{lemma:perfect-hash}
For any $n, k \geq 1$ one can deterministically construct an $(n, k)$-family of perfect hash functions of size $e^{k}k^{O(\log k)} \log n$ in time $e^{k}k^{O(\log k)}n \log n$ in the centralized setting.
\end{lemma}

We can now use families of perfect hash functions to derandomize Theorem~\ref{thm:path}, assuming that all nodes know the value $\degen(G)$.

\begin{theorem}\label{thm:derandomisation}
Finding induced copy of a tree $T$ on $k$ nodes in a $\degen$-degenerate graph $G$ can be done in $f(\degen,k) + O(\log n)$ rounds in the \congest model using a deterministic algorithm for some function $f$, assuming $\degen$ is known to all nodes.
\end{theorem}

\begin{proof}
Let $N$ be the polynomial upper bound the nodes have for the number of nodes in the graph $G$, and let $\sigma$ be the $\alpha$-orientation of $G$ for $\alpha = 3\degen$. We want to construct a family of colorings $\F$ of $V$ so that if $G$ contains an induced copy $H$ of the pattern graph, then there is some $\chi \in \F$ such that $H$ is proper w.r.t.~$\sigma$ and $\chi$. The nodes construct the desired $\F$ as follows:
\begin{enumerate}
    \item Each node $v$ locally constructs an $(N,k+\degen k)$-family of perfect hash functions $\F'$ from $[N]$ to $C = \{ 1, 2, \dotsc, k + \alpha k \}$. This is done deterministically based on $N$, $\alpha$ and $k$, so each node has the same $\F'$.
    \item For every $f \in \F'$, subset $L \subseteq C$ of size $k$, and bijection $b \colon L \to \{ 1, 2, \dotsc, k \}$, each node $v$ adds the function $\chi \colon [N] \to \{ 1, 2, \dotsc, k \}$ defined by
\[ \chi(i) =
\begin{cases}
(g \circ f)(i) & \text{ if $f(i) \in L$, and}\\
0 & \text{ otherwise,}
\end{cases} \]
to the family $\F$.
\end{enumerate}
By construction, we have that for any disjoint sets $S, T \subseteq [N]$ with $|S| = k$ and $|T| = \alpha k$ and  any bijection $b \colon S \to [k]$ there is $\chi \in \F$ such that $\chi$ agrees with $b$ on $S$ and $T$ is colored entirely with color $0$. 

To derandomize the algorithm of Theorem~\ref{thm:path}, the nodes construct the family of colorings $\F$ as above and use each function $\chi \in \F$ to select their colors as $\chi(v)$ in turn instead of the random choice. The total number of function in $\F$ is $s(\degen,k) \log N = O(s(\degen,k)\log n)$ for some function $s$, the running time is $s(\degen,k) + O(\log n)$.
\end{proof}

\paragraph{Unknown degeneracy.} The only part where the randomized algorithm uses the knowledge of $\degen(G)$ is for deciding how many repeats of the random coloring it performs; Lemma~\ref{lemma:compute-orientation} can be used without knowing $\degen(G)$. Without knowledge of $\degen(G)$, nodes can determine the largest out-degree in orientation $\sigma$ in their radius-$k$ neighborhood and use that as a proxy for $\degen(G)$ to determine how many repeats of the random coloring they should participate in; it is easy to verify that this still retains the correctness of the algorithm. The only caveat is that different nodes can terminate at different times, and cannot determine when all nodes have terminated.

The deterministic algorithm, on the other hand, requires that all nodes know the degeneracy $\degen(G)$, or the same upper bound for this value. While we can compute an $O(\kappa(G))$-orientation $\sigma$ for $G$ in $O(\log n)$ rounds, all nodes do not necessarily learn the largest out-degree in $\sigma$; indeed, one can trivially see that having all nodes learn $\degen(G)$ requires $\Omega(D)$ rounds in the worst case.

\subsection{Induced subgraph detection for bounded treewidth patterns}\label{sec:lb-low-treewidth-degen}

We now show that with slight modification, the hard treewidth $2$ patterns presented in Section~\ref{sec:lb-low-treewidth} can be adapted to bounded degeneracy setting. Recall that as mentioned in the introduction, any pattern graph on $k$ nodes can be detected in $O(k \degen(G) n)$ rounds by having all nodes gather full information about their distance-$k$ neighborhood; thus, the following lower bound is almost tight.

\begin{theorem}
For any $k \ge 2$, there exists a pattern graph $H_k$ of treewidth $2$ such that \congest algorithm solving either $H$-detection or induced $H$-detection on graphs of degeneracy $2$ needs at least $\Omega(n^{1-1/k})$ rounds.
\end{theorem}

We use the same construction for $k \ge 2$ for the pattern graph as in Lemma~\ref{thm:bounded-treewidth-pattern}, but add paths of length $5$ instead of paths of length $3$ between triangles $A_1$ and $B_1$, and triangles $A_2$ and $B_2$. Let us denote the resulting graph by $H'_k$.

\begin{center}
\includegraphics[width=0.6\textwidth]{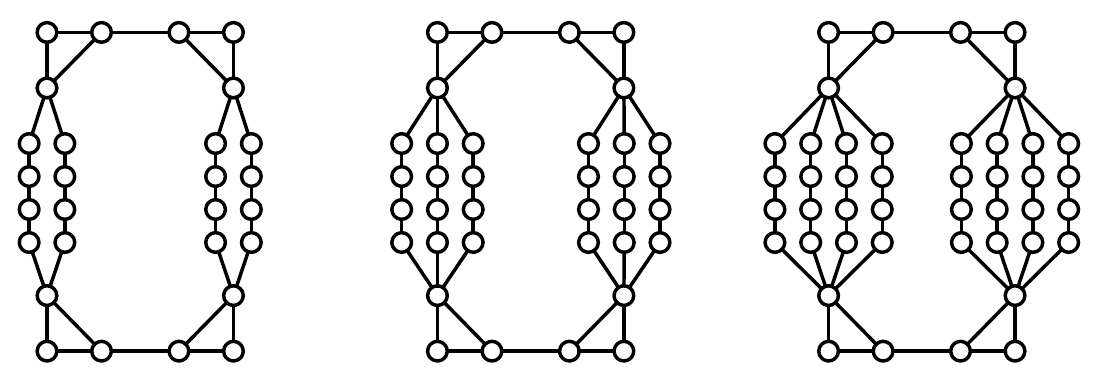}
\end{center}

\begin{lemma}
Let $k \ge 2$ be fixed. There exists a family of lower bound graphs of degeneracy $2$ for $H'_k$-detection and induced $H'_k$-detection with $f_n = \disj_{\Theta(n)}$ and $C(n) = \Theta(n^{1/k})$.
\end{lemma}

\paragraph{Construction.} For the family of lower bound graphs, we also use a variation of the construction for Lemma~\ref{lemma:bounded-treewidth-lb}. Let $N \in \mathbb{N}$ and $x_0, x_1 \in \{ 0, 1 \}^{N}$. We construct a graph $G(n,x_0,x_1)$ as follows:
\begin{enumerate}
    \item We start with the same node set as in Lemma~\ref{lemma:bounded-treewidth-lb}.
    \item Instead of the edges described in the proof of Lemma~~\ref{lemma:bounded-treewidth-lb}, the graph $G(n,x_0,x_1)$ has the following additional edges for all $x_0$ and $x_1$:
    \begin{enumerate}[label=(\alph*),noitemsep]
        \item For each $j \in \{ 1, 2, \dotsc, K \}$, there are edges $\{ a_{1,j}, b_{1,j}\}$ and $\{ a_{2,j}, b_{2,j}\}$ (same as Lemma~\ref{lemma:bounded-treewidth-lb}.)
        \item Fix an injection $\rho$ from $[N]$ to subsets of $[K]$ of size $k$. For each $i \in \{ 1,2 \}$ and $j \in [N]$, we add a \emph{path of length $2$} between node $2$ of triangle $A_{i,j}$ and $a_{i,\ell}$ for each element $\ell \in \rho(j)$, and likewise between between node $2$ of triangle $B_{i,j}$ and $b_{i,\ell}$.
    \end{enumerate}
    \item Depending on $x_0$ and $x_1$, we add the following edges to the graph. For each $j \in [N]$ such that $x_0(j) = 1$, we add the edge between node $1$ of triangle $A_{1,j}$ and node $1$ of triangle $A_{2,j}$. Likewise, for each $j \in [N]$ such that $x_1(j) = 1$, we add the edge between nodes $1$ of triangles $B_{1,j}$ and $B_{2,j}$.
\end{enumerate}
Figure~\ref{fig:bounded-treewidth-lb-degen} shows an example of this construction.

The total number of nodes in the graph $G(n,x_0,x_1)$ is $n = 12N + 4 kN^{1/k} + kN = \Theta(N)$. The partition required by Definition~\ref{def:lower-bound-family} is obtained by setting $V_0$ to contain all triangles $A_{i,j}$ and nodes $a_{i,\ell}$. The size of the cut $S$ is $kN^{1/k} = \Theta(n^{1/k})$.

To see that $G(n,x_0,x_1)$ has degeneracy $2$, we observe that any node in $G(n,x_0,x_1)$ has at most $2$ neighbors of degree larger than $2$. Thus, any induced subgraph $G'$ of $G(n,x_0,x_1)$ either contains a node with degree $2$ in $G(n,x_0,x_1)$, or every node in $G'$ has degree at most $2$.

\begin{figure}
\begin{center}
\includegraphics[width=0.5\textwidth]{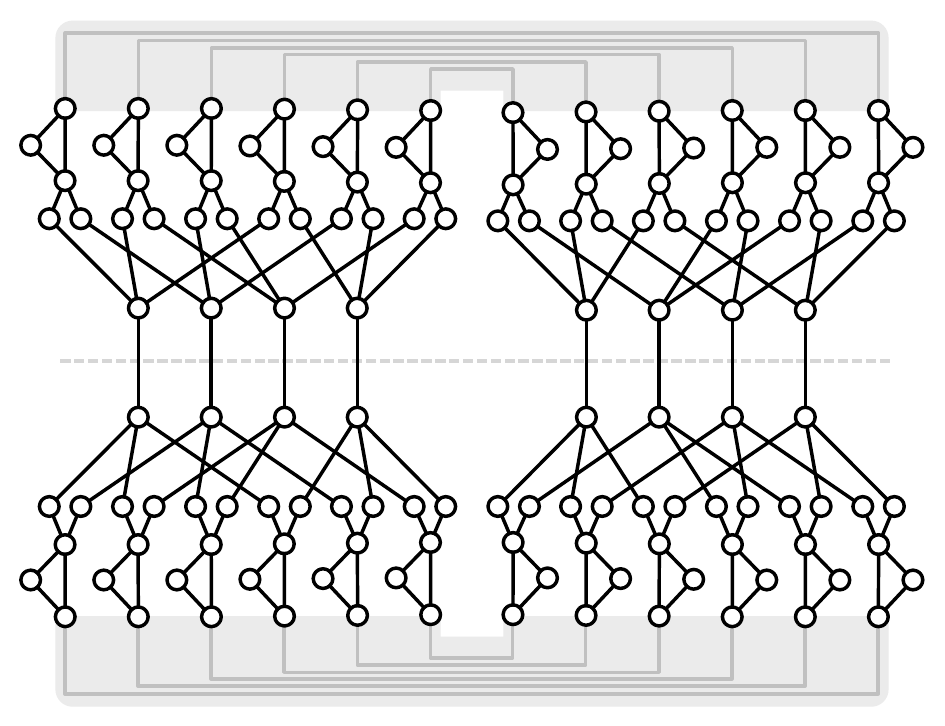}
\caption{Example of the hard instance for $H'_k$-detection for $k = 2$. The shaded area is represents the encoding edges.}\label{fig:bounded-treewidth-lb-degen}
\end{center}    
\end{figure}

\begin{center}

\end{center}

\begin{lemma}
Let $G(n, x_0, x_1)$ be the graph defined above. Then we have that $G(n, x_0, x_1)$ contains a copy of $H'_k$ if and only if $x_0$ and $x_1$ are not disjoint. Moreover, any copy of $H'_k$ in $G(n, x_0, x_1)$ is an induced copy of $H'_k$.
\end{lemma}

\begin{proof}
Argument is identical to Lemma~\ref{lemma:bounded-treewidth-lb-aux1}.
\end{proof}

\section{Bounded vertex cover number and MCIS}\label{sec:mcis}

Finally, we consider induced subgraph detection parameterized by vertex cover number $\vc(G)$. Specifically, we show that a more general problem of \emph{maximum common induced subgraph (MCIS)} can be solved in constant rounds on graphs of constant vertex cover number, which implies our results for induced subgraph detection.

\subsection{Preliminaries}

\paragraph{Maximum common induced subgraph.}

In the centralized version of maximum common induced subgraph, we are given graphs $G = (V_G, E_G)$ and $H = (V_H, E_H)$ as input, and the task is to find the maximum-size graph $G^*$ such that $G^*$ appears as induced subgraph of both $G$ and $H$. More precisely, the output should be a function $f \colon V_G \to V_H \cup \{ \bot \}$ such that for the set $U_G = \{ v \in V_G \colon f(v) \ne \bot \}$, the function $f$ restricted to $U_G$ is an isomorphism between $G[U_G]$ and $H[f(U_G)]$.

In this section, we consider MCIS parameterized by the sum of the vertex cover numbers $\vc(G) + \vc(H)$. Note that when $H$ is a complete graph and $|G| = |H|$, the problem is equivalent to maximum clique and hence NP-hard. It is W[1]-hard when parameterized by the solution size $k$, and W[1]-hard parameterized by the size of a minimum vertex cover of only one of the input graphs, even when restricted to bipartite graphs (see e.g.~\cite{ABUKHZAM201499,ABUKHZAM201769} for more discussion).

\paragraph{Distributed MCIS.} In the distributed version of MCIS, the input graph $G = (V_G, E_G)$ is the communication network, and full information about the second input graph $H = (V_H, E_H)$ is given to every node as local input. Each node $v$ needs to give a local output $f(v) \in V_H \cup \{ \bot \}$ such that the global function $f$ satisfies the conditions of MCIS solution.

A key observation for our distributed algorithm is the following trivial lemma; a graph of bounded vertex cover number also has bounded diameter, which allows us to perform global coordination easily.

\begin{lemma}\label{lemma:vc-diameter}
If graph $G$ has vertex cover number $\vc(G) \le k$, then $G$ has diameter at most $2k$.
\end{lemma}

\subsection{Centralized algorithm}

We first present the centralized algorithm of \citet{ABUKHZAM201769}, and subsequently show how to implement it in the distributed setting. The algorithm is a fairly straightforward enumeration of possible solutions.

Let $\vc = \max(\vc(G),\vc(H))$ be an upper bound for both of the vertex cover numbers. We denote by $C_G$ and $I_G$ the minimum vertex cover of $G$ and its complement, respectively; note that the latter is an independent set. Likewise, we denote by $C_H$ and $I_H$ the minimum vertex cover of $H$ and its complement. By definition, we have $|C_G|, |C_H| \le \vc$.

The MCIS algorithm of \citet{ABUKHZAM201769} performs the following steps to find a maximum common induced subgraph:
\begin{enumerate}
	\item We first enumerate all ways to partition $C_G$ into three sets $\mcistoC{G}$, $\mcistoI{G}$, and $C_{G, \to\bot}$ and $C_H$ into three sets $\mcistoC{H}$, $\mcistoI{H}$, and $C_{H, \to\bot}$. For $\Gamma \in \{ G, H \}$, the nodes in $\mcistoC{\Gamma}$ are to be mapped to the vertex cover in the other graph, nodes in $\mcistoI{\Gamma}$ are to be mapped to the independent set in the other graph, and nodes in $C_{G, \to\bot}$ are to be mapped to $\bot$. Partitionings where $|\mcistoC{G}| \ne |\mcistoC{H}|$ are ignored. There are $9^\vc$ possible ways to select these partitionings.
	\item For $U \subseteq \mcistoC{G} \cup \mcistoI{G}$, we define
	\[ I_{G,U} = \{ v \in I_G \colon N(v) \cap \mcistoC{G} \cup \mcistoI{G}  = U  \}\,.\]
	This partitions $I_G$ into equivalence classes with the same neighborhoods w.r.t.\ nodes that are not mapped to $\bot$. Similarly, we define equivalence classes $I_{H,W}$ for $W \subseteq \mcistoC{H} \cup \mcistoI{H}$. Note that there are at most $2^\vc$ equivalence classes in either graph.
	\item For each partition in Step~(1), we enumerate all bijective mappings from $\mcistoC{G}$ to $\mcistoC{H}$, mappings from $\mcistoI{G}$ to equivalence classes of $I_H$, and mappings from $\mcistoI{H}$ to equivalence classes of $I_G$. There are at most $\vc^\vc$ ways to select the first mapping, and at most $2\cdot (2^\vc)^\vc$ ways to select the latter two mappings.
	\item For each set of mappings selected in Step~(3), we check that each equivalence class $I_{G,U}$ (resp. $I_{H,W}$) is an image of at most $|I_{G,U}|$ elements in $\mcistoI{H}$ (resp. at most $|I_{H,W}|$ elements in $\mcistoI{G}$). If this holds, we then construct a partial mapping $f \colon V_G \to V_H$ according to mappings selected in Step~(3), by selecting arbitrarily representatives from the equivalence classes, and the check that $f$ is a valid partial isomorphism.
	\item Finally, to complete the mapping $f$, we select for each remaining unmapped node $v \in I_{G,U}$ an image from the corresponding equivalence class in $I_H$ -- i.e. from the equivalence class $I_{H,W}$ that satisfies $f(U) = W$ -- as long as there remain free nodes. The remaining nodes $v \in V_G$ are mapped to $\bot$.
\end{enumerate}
Throughout the enumeration process, the algorithm keeps track of the best solution found so far, and returns the best solution upon termination of the algorithm. It follows immediately that any function $f$ enumerated by this algorithm is a valid MCIS solution, and an optimal solution will be listed by the algorithm.

\subsection{Distributed implementation}

We now show how to implement the algorithm of \citet{ABUKHZAM201769} in the \congest model.

\begin{theorem}\label{thm:dist-mcis}
Solving the maximum common induced subgraph problem on communication graph $G$ and target graph $H$ can be done in $2^{O(\vc^2)}$ rounds in the \congest model deterministically, where $\tau = \max(\vc(G),\vc(H))$.
\end{theorem}

\begin{proof}
First, using standard techniques, we can construct a BFS tree $T$ of depth $O(D)$ in $O(D)$ rounds~\cite{Peleg00}, and since $D = O(\tau)$ by Lemma~\ref{lemma:vc-diameter}, we get a BFS tree of depth $O(\vc)$ in $O(\vc)$ rounds.

We can compute a minimum vertex cover $C_G$ of $G$ in $O(\vc^2 \log \vc)$ rounds using the algorithm of \citet{ben2018parameterized} and binary search, and since all nodes know the graph $H$, we can also compute a minimum vertex cover $C_H$ for $H$ locally at each node. Broadcasting over the BFS tree $T$, we can have all nodes learn the identifiers of all nodes in $C_G$ in $O(\vc)$ rounds. Moreover, we assume that all nodes learn the identifiers of their neighbors and whether they belong to $C_G$ at this point.

We now describe how to implement the steps of the centralized MCIS algorithm in \congest.

\noindent\textbf{Implementing Step~(1).} All nodes know the sets $C_G$ and $C_H$ in full, so all nodes can list the possible partitionings locally. The subsequent steps are performed for one fixed partitioning at a time, iterating over all possibilities.

\noindent\textbf{Implementing Step~(2).} All nodes now know the current partitioning of $C_G$ and $G_H$. Thus, each node in $I_G$ can compute their own equivalence class, and all nodes can compute all equivalence classes of nodes in $I_H$.

\noindent\textbf{Implementing Step~(3).} All nodes know the current partitioning of $C_G$, so each node knows the possible equivalence classes of $I_G$; however, note that nodes do not know their sizes at this step, and some may be empty. By the previous step, all nodes know the exact partitioning of $I_H$ into equivalence classes. Based on this information, all nodes can enumerate the possible mappings locally. The subsequent steps are performed for one fixed partitioning and one fixed set of mappings at a time, iterating over all possibilities.

\noindent\textbf{Implementing Step~(4).} To check if the mapping $f$ can be constructed, the nodes need to learn the sizes of the equivalence classes of $I_G$ and $I_H$. The latter is known to all nodes per the previous step. For the former, we use standard gathering and broadcast over the BFS tree $T$ to learn the size of each $I_{G,I}$ for $U \subseteq \mcistoC{G} \cup \mcistoI{G}$; this takes $O(\vc 2^\vc)$ rounds. Each node can then perform the necessary checks locally.

Each node $v$ in $\mcistoC{G}$ locally fixes their output $f(v)$ according to the mapping selected in Step~(3). Each node $v$ in $\mcistoI{G}$ selects an output $f(v)$ from the equivalence class $I_{H,W}$ they are mapped to; in more detail, since $v$ knows the identifiers of all nodes in $\mcistoI{G}$ that are mapped to the same equivalence $I_{H,W}$, it can compute index $i$ such that $v$ is the node with $i$th largest identifier among the nodes mapped to $I_{H,W}$. Node $v$ then selects $w \in I_{H,W}$ that has the $i$th largest identifier in set $I_{H,W}$ as the image $f(v)$.

To fix the remaining outputs assigned at this step, we first, for each equivalence class $I_{G,U}$, assign indices $i \in \{ 1, 2, \dotsc, |I_{G,U}| \}$ to each node $v \in I_{G,U}$ so that each node $v$ in the class knows its own index and has a distinct index. This can be done by e.g.~assigning indices for each class separately in DFS order over the BFS tree $T$, taking total time $O(\vc 2^\vc)$ for all equivalence classes. Using this numbering, all nodes can locally compute a mapping from the nodes in $\mcistoI{H}$ to the indices of the equivalence classes as above. Each node $v \in I_{G}$ can locally determine if they are paired with a node $w \in \mcistoI{H}$, and set their local output as $f(v) = w$.

Finally, the nodes learn the assigned output values $f(v)$ of their neighbors, and check that $f$ is a valid partial isomorphism.

\noindent\textbf{Implementing Step~(5).} At this point, all nodes know how many unassigned nodes are left in each equivalence class $I_{G,U}$ and $I_{H,W}$, and that the assigned nodes are the ones with smallest indices (in case of $I_{G,U}$) and smallest identifiers (in case of $I_{H,W}$). For each pair of $U \subseteq \mcistoC{G} \cup \mcistoI{G}$ and $W \subseteq \mcistoC{H} \cup \mcistoI{H}$ with $f(U) = W$, we assign as many of the remaining nodes from the equivalence class $I_{G,U}$ as possible to be mapped to nodes in $I_{G,W}$ in the identifier order. As each node $v \in I_G$ knows their index within their equivalence class, they can locally determine if they are to be mapped to a node $w \in V_G$, and output $f(v) = w$. All remaining nodes $v \in V_G$ output $f(v) = \bot$.

\noindent\textbf{Running time.} There is at most total of $9^\vc \vc^\vc 2 \cdot 2^{\vc^2}$ selections of the partitions and mappings that are checked by the algorithm, and checking one combination can be done in $O(\vc 2^\vc)$ rounds. The total running time is thus $O(18^\vc \vc^{\vc + 1} 2^{\vc^2 + 1})$ rounds.
\end{proof}

\subsection{Induced subgraph detection on bounded vertex cover number graphs}

As an immediate consequence of the MCIS algorithm, we obtain a parameterized distributed algorithm for detecting an induced copy of $H$, for any pattern graph $H$. That is, any graph $H$ on $k$ nodes has vertex cover number at most $k$, and the optimal MCIS solution for graph $G$ and $H$ has size $k$ if and only if $H$ appear as an induced subgraph in $G$, assuming $G$ has more nodes than $H$.

\begin{theorem}\label{thm:subgraph-vcn}
Let $H$ be a pattern graph on $k$ nodes. Finding induced copy $H$ can be done in $2^{O((\vc(G) + k)^2)}$ rounds in the \congest model deterministically.
\end{theorem}

\section{Conclusions and open problems}

A central takeaway of this work is that centralized parameterized complexity offers both algorithmic techniques and perspectives for distributed computing. In particular, we believe that the study of structural graph parameters is a valuable paradigm for understanding sparse and structured networks in general. However, we note that there still remain open research directions related to topics studied in this paper:
\begin{itemize}
	\item In terms of immediate open questions left by our work, we note that we currently do not have any systematic results on separation between the hardness of induced and non-induced subgraph detection for a given pattern $H$. For example, the induced cycle detection lower bound of Le Gall and Miyamoto~\cite{legall2021induced} gives a near-linear -- or super-linear, in case of even cycles -- gap between induced and non-induced cycle detection, but it would be interesting to explore similar results for other pattern graphs in systematic fashion.
	\item More generally, we do not understand the complexity of subgraph detection type problems in the distributed setting as well as in the centralized setting. For example, the complexity of $k$-independent set detection in \congest remains open, whereas in the centralized setting, it is equivalent to $k$-clique -- a correspondence that does not hold in \congest. 
	\item Besides degeneracy and vertex cover number, there are many other structural graph parameters commonly studied in parameterized complexity -- for example, feedback vertex and edge sets, treewidth, and pathwidth. Whereas Li~\cite{li2018distributed} provides a framework for using treewidth for global optimization problems, it does not directly imply results for local problems such as subgraph detection; one might expect that considering something akin to \emph{local} treewidth of a graph would be more appropriate for local graph problems. A secondary question is understanding what structural graph parameters are relevant from the perspective of real-world networks.
\end{itemize}

\subsection*{Acknowledgements}

Project has been supported by the European Research Council (ERC) under the European Union’s Horizon 2020 research and innovation programme (grant agreement No 805223 ScaleML), and by the LABEX MILYON (ANR-10-LABX-0070) of Universit\'e de Lyon, within the program ``Investissements d'Avenir'' (ANR-11-IDEX-0007)  operated by the French National Research Agency (ANR). We thank Fran\c{c}ois Le Gall and Masayuki Miyamoto for sharing their work on lower bounds for induced subgraph detection~\cite{legall2021induced}.

\DeclareUrlCommand{\Doi}{\urlstyle{same}}
\renewcommand{\UrlFont}{\footnotesize\sf}
\renewcommand{\doi}[1]{\href{http://dx.doi.org/#1}{\footnotesize\sf doi:\Doi{#1}}}

\bibliographystyle{plainnat}
\bibliography{congest-parameterised}

\clearpage

\appendix

\section{Induced short paths}\label{sec:induced-short-paths}

Induced paths with two edges can be detected in $O(1)$ rounds, in contrast to the situation with e.g.~triangle detection. The proof follows the centralized algorithm of \citet{vassilevska2008efficient}.

\begin{theorem}
Given a graph $G$ on $n$ nodes, detecting an induced path of length 2 on $G$ can be done in $O(1)$ rounds in the broadcast \congest model.
\end{theorem}

\begin{proof}
As the first step, we assign a label $\ell(v)$ for each node as follows. First, each node $v \in V$ broadcast its identifiers to all its neighbors $N(v)$, and then each node $v$ picks the label $\ell(v)$ to be the smallest identifier from the set that it received, or its own identifier if that is smaller. The nodes then broadcast their label $\ell(v)$ and their degree $\deg(v)$ to all their neighbors.

Each node $v$ then checks the following conditions, and reports that induced $2$-path exists if at least one of them is satisfied:
\begin{enumerate}[label=(\alph*)]
    \item The exists a neighbor $u \in N(v)$ with $\deg(v) \ne \deg(v)$.
    \item There exists neighbors $u, w \in N(v)$ with $\ell(u) \ne \ell(w)$.
\end{enumerate}

For the correctness of the algorithm, we first observe that a graph does not contain an induced $2$-path if and only if each connected component is a clique. If none of the nodes report an induced $2$-path, then by conditions (a) and (b), each connected component is a clique. Likewise, if $G$ consists of disjoint cliques, no node will report an induced $2$-path. Finally, we note that the algorithm takes 3 rounds in \congest.
\end{proof}

\end{document}